\documentclass{article}
\usepackage{geometry}
\usepackage[nospace,noadjust]{cite}
\usepackage{amsmath,amssymb,amsfonts,times}
\usepackage{graphicx}
\usepackage{textcomp}
\usepackage{xcolor}
\usepackage{subfigure}
\usepackage{epsfig}
\usepackage{multirow}
\usepackage{algorithm}
\usepackage{algorithmicx}
\usepackage{algpseudocode}
\usepackage{array}
\usepackage{epstopdf}
\usepackage{color}
\usepackage{diagbox}
\usepackage{bm}
\usepackage{bbm}
\usepackage{mathrsfs}
\usepackage{tcolorbox}
\usepackage{pdfsync}

\usepackage{url}
\usepackage[hidelinks]{hyperref}
\usepackage{amsmath,amsthm,amssymb,amsbsy}
\usepackage{paralist}
\usepackage{xcolor}
\usepackage{color}
\usepackage{comment}
\usepackage{multirow}
\usepackage[toc, page]{appendix}

\usepackage{fancyhdr}
\usepackage{cite}
\usepackage{cleveref}

\newtheorem{lemma}{Lemma}
\newtheorem{defi}{Definition}

\newtheorem{theorem}{Theorem}

\theoremstyle{remark}

\newcommand{\R}{\mathbb{R}}
\def \real    { \mathbb{R} }

\newcommand{\C}{\mathbb{C}}

\newcommand{\e}{\begin{equation}}
\newcommand{\ee}{\end{equation}}
\newcommand{\en}{\begin{equation*}}
\newcommand{\een}{\end{equation*}}
\newcommand{\eqn}{\begin{eqnarray}}
\newcommand{\eeqn}{\end{eqnarray}}
\newcommand{\bmat}{\begin{bmatrix}}
\newcommand{\emat}{\end{bmatrix}}

\DeclareMathAlphabet\mathbfcal{OMS}{cmsy}{b}{n}
\renewcommand{\P}[1]{\operatorname{\mathbb{P}}\left(#1\right)}

\newcommand{\vct}[1]{\boldsymbol{#1}}
\newcommand{\mtx}[1]{\boldsymbol{#1}}

\newcommand{\<}{\langle}
\renewcommand{\>}{\rangle}

\newcommand{\trace}{\operatorname{trace}}

\newcommand{\rank}{\operatorname{rank}}

\newcommand{\set}[1]{\mathbb{#1}}

\DeclareMathOperator*{\argmin}{\text{arg~min}}

\newcommand{\wh}{\widehat}

\newcommand{\wt}{\widetilde}
\newcommand{\ol}{\overline}
\newcommand{\ul}{\underline}
\newcommand{\nqbit}{n}

\newcommand{\innerprod}[2]{\left\langle #1,  #2 \right\rangle}

\newcommand{\vf}{\vct{f}}

\newcommand{\vp}{\vct{p}}

\newcommand{\vpsi}{\vct{\psi}}

\newcommand{\veta}{\vct{\eta}}

\newcommand{\vrho}{\vct{\rho}}

\newcommand{\mA}{\mtx{A}}
\newcommand{\mB}{\mtx{B}}

\newcommand{\mF}{\mtx{F}}

\newcommand{\mX}{\mtx{X}}

\newcommand{\mId}{{\bf I}}

\newcommand{\setF}{\set{F}}

\newcommand{\setX}{\set{X}}

\graphicspath{{./figs/}}

\newlength{\imgwidth}
\newboolean{twoColVersion}
\setboolean{twoColVersion}{false}
\newcommand{\twoCol}[2]{\ifthenelse{\boolean{twoColVersion}} {#1} {#2} }

\usepackage{mathtools}

\title{\LARGE \bf A Unified Framework for Sample Complexity of Structured \\ Quantum State Tomography under Noisy Observations}

\author{Zhen Qin\thanks{Zhen Qin is with the Michigan Institute for Computational Discovery and Engineering, Department of Electrical Engineering and Computer Science and Department of Statistics, University of Michigan, Ann Arbor, MI 48109 USA. (e-mail: zhenqin@umich.edu).}}

\begin{document}

\maketitle

\begin{abstract}
Quantum state tomography (QST) has attracted considerable attention due to its fundamental role in quantum information processing. Although substantial progress has been made in reducing the sample complexity of structured QST, the current theoretical understanding remains fragmented across different structured quantum-state classes and reconstruction models. Moreover, existing sample complexity analyses almost exclusively assume ideal measurements and therefore do not account for realistic state preparation and measurement noise. In this paper, we develop a unified theoretical framework for analyzing the sample complexity of structured QST under noisy observations arising from state preparation noise, measurement noise, and finite-shot statistical noise. The proposed framework applies to a broad family of structured quantum-state classes, including general mixed states, sparse states, low-rank states, matrix product states (MPSs), matrix product operators (MPOs), projected entangled-pair states (PEPSs), and projected entangled-pair operators (PEPOs), while further introducing physically consistent structured models---including low-rank and sparse states, low-rank MPOs (LR-MPOs), and low-rank PEPOs (LR-PEPOs)---that simultaneously exploit low-dimensional structures and preserve the physical constraint. Within this framework, we derive unified non-asymptotic sample complexity guarantees for two constrained least-squares estimators under noisy observations: a noise-aware estimator that incorporates the calibrated noise model and a noise-unaware estimator based on the ideal Born measurement model. For the noise-aware estimator, we derive unified trace-norm recovery guarantees that explicitly characterize the dependence of the sample complexity on three fundamental quantities: the complexity of the underlying structured state class, the complexity of the measurement ensemble, and the state preparation and measurement noise levels. For the noise-unaware estimator, we establish a unified non-asymptotic recovery guarantee consisting of a statistical error term and an additional deterministic bias term arising from the mismatch between the assumed reconstruction model and the noisy observation process. While the statistical error preserves the same dependence on the structured state complexity and the measurement complexity as in the noise-aware setting, the deterministic bias is independent of the number of state copies and therefore cannot be eliminated by increasing the measurement budget.
\end{abstract}


\section{Introduction}
\label{sec:introduction}

Quantum state tomography (QST) provides a fundamental framework for characterizing unknown quantum states and plays a central role in quantum information processing, quantum computing, and quantum sensing~\cite{bertrand1987tomographic,vogel1989determination,leonhardt1995quantum,hradil1997quantum,james2001measurement}. For an $n$-qudit quantum system, where each qudit has local dimension $d$, the quantum state is represented by a density matrix $\vrho\in\mathbb{C}^{d^n\times d^n}$. Without any structural assumption, accurately estimating $\vrho$ requires a number of state copies that scales exponentially with the system size. In particular, under independent measurement schemes, achieving a bounded reconstruction error in trace norm requires at least $O(d^{3n})$ state copies~\cite{haah2017sample}. Such exponential sample complexity rapidly becomes prohibitive for modern quantum devices, whose system sizes have already exceeded one hundred qubits~\cite{preskill2018quantum,arute2019quantum,chow2021ibm}.

To overcome this exponential bottleneck, a substantial body of recent work has focused on exploiting low-dimensional structures that naturally arise in physically relevant quantum systems. Typical examples include pure or nearly pure states~\cite{bartlett2007reference}, low-temperature thermal states~\cite{hartmann2004existence}, a family of generalized GHZ states~\cite{jameson2024optimal}, one-dimensional quantum many-body systems~\cite{eisert2010}, systems with decaying long-range interactions~\cite{pirvu2010matrix}, and states generated by noisy intermediate-scale quantum devices~\cite{noh2020efficient}. These structural assumptions substantially reduce the effective degrees of freedom of the underlying quantum state and consequently improve the sample complexity of QST. Among the various structured models, low-rank density matrices and tensor-network representations have emerged as two of the most widely studied paradigms. For low-rank quantum states with $\rank(\vrho)=r$, the required number of state copies can be reduced to $O(d^n r^2)$ while maintaining rigorous reconstruction guarantees~\cite{flammia2012quantum,voroninski2013quantum,haah2017sample,kueng2017low,guctua2020fast,francca2021fast,qin2026optimal}. Although this represents an exponential improvement over unstructured tomography, the dependence on the Hilbert-space dimension $d^n$ remains significant for large-scale quantum systems. To overcome this limitation, tensor-network representations exploit locality and limited entanglement, replacing exponential complexity with polynomial scaling for many physically relevant quantum states. Building on these representations, a growing body of work has established polynomial sample complexity guarantees for tensor-network quantum state tomography, reducing the required number of state copies to $O(\textup{poly}(n))$ while achieving accurate recovery in Frobenius norm~\cite{qin2024quantum,qin2024sample,qin2025enhancing,tang2025sketch,votto2026learning,qin2026quantum,cai2026online}.

Despite these advances, the current theoretical understanding of structured QST remains incomplete in several important aspects. As summarized in the recent review~\cite{qin2026statistical}, existing theoretical guarantees are typically established for specific reconstruction algorithms and individual state models, making it difficult to compare different structural assumptions under a common theoretical framework. Furthermore, most existing analyses of tensor-network tomography are developed under the Frobenius norm, whereas the trace norm is the standard metric for quantifying the distinguishability of quantum states. Moreover, the physical constraint has received relatively limited theoretical attention in tensor-network tomography, and several practically important structured state classes, such as sparse quantum states, have not yet been systematically investigated from a sample complexity perspective.

Beyond these limitations in modeling structured state classes, existing sample complexity analyses almost exclusively account for finite-shot statistical noise while assuming an ideal measurement process. In practical quantum devices, however, state preparation and measurement are inevitably corrupted by hardware imperfections. Depolarizing noise has consequently become one of the standard models for characterizing state preparation and measurement imperfections in practical QST~\cite{gross2010quantum,farooq2022robust,ivanova2023optimal,aasen2024readout,winderl2024quantum,de2024quantum}. Despite its practical importance, the theoretical understanding of sample complexity under depolarizing noise remains very limited. Existing theoretical results are largely restricted to noisy Pauli shadow tomography~\cite{cotler2026noisy}, whereas comparable guarantees for QST are still lacking. These observations naturally lead to the following fundamental question.

\begin{tcolorbox}[colback=white,left=1mm,top=1mm,bottom=1mm,right=1mm,boxrule=.3pt]
\centering
{\bf Question}: Can one develop a unified theoretical framework that characterizes the sample complexity of structured QST under realistic noisy observations?
\end{tcolorbox}

In this paper, we answer this question affirmatively. Specifically, we establish a unified theoretical framework for analyzing the sample complexity of structured QST under realistic noisy observations arising from state preparation noise, measurement noise, and finite-shot statistical noise. Our framework applies to a broad family of structured quantum-state classes, including general mixed states, sparse states, low-rank states, matrix product states (MPSs), matrix product operators (MPOs), projected entangled-pair states (PEPSs), and projected entangled-pair operators (PEPOs). To simultaneously exploit low-dimensional structures and preserve the physical constraint, we further propose several physically consistent structured models, including low-rank and sparse states, low-rank MPOs (LR-MPOs), and low-rank PEPOs (LR-PEPOs). A unified covering-number analysis is developed for all considered structured state classes, providing a common characterization of their intrinsic complexities.

Building upon this framework, we establish unified non-asymptotic sample complexity guarantees for two constrained least-squares estimators under noisy observations: a noise-aware estimator that incorporates the known noise model and a noise-unaware estimator that ignores the noise effects. For the noise-aware estimator, which explicitly incorporates the calibrated state preparation and measurement noise into the reconstruction model, we derive unified trace-norm recovery guarantees for all considered structured quantum-state classes. The resulting sample complexity depends on three fundamental quantities: the logarithm of the covering number of the underlying structured state class, which captures its intrinsic degrees of freedom; a normalized measurement complexity parameter, which characterizes the statistical efficiency of the measurement ensemble; and an explicit multiplicative factor determined by the state preparation and measurement noise levels. Moreover, our framework provides explicit evaluations of the measurement complexity parameters for spherical $3$-designs, $\delta$-approximate spherical $3$-designs, unitary $3$-designs, and Haar-random projective measurements, thereby quantitatively comparing the statistical efficiency of different measurement ensembles.

We further investigate the practically important noise-unaware estimator, which reconstructs the quantum state using the ideal Born measurement model without explicit knowledge of the state preparation and measurement noise. For this statistically misspecified estimator, we establish a unified non-asymptotic reconstruction guarantee showing that the estimation error naturally decomposes into a statistical error term and an additional deterministic bias term. While the statistical error term maintains a similar dependence on the covering number and measurement complexity as in the noise-aware setting, the noise-unaware estimator introduces an additional deterministic bias term that does not vanish with increasing numbers of state copies. This term captures the error floor caused by the mismatch between the assumed reconstruction model and the actual noisy observation process, and is determined by the state preparation and measurement noise levels as well as the measurement properties. Consequently, the proposed analysis provides a quantitative characterization of the fundamental performance gap between noise-aware and noise-unaware structured QST under realistic noisy observations.

\subsection{Notation}
\label{sec:notation}

We use bold capital letters (e.g., $\bm{X}$) to denote matrices,  bold lowercase letters (e.g., $\bm{x}$) to denote column vectors, and italic letters (e.g., $x$) to denote scalar quantities.  For a positive integer $K$, $[K]$ denotes the set $\{1,\dots, K \}$. The superscripts $(\cdot)^\top$ and $(\cdot)^\dagger$ denote the transpose and Hermitian transpose, respectively. For two matrices $\mA,\mB$ of the same size, $\innerprod{\mA}{\mB} = \trace(\mA^\dagger\mB)$ denotes the inner product between them.
For two positive quantities $a,b\in \real$, the inequality $b\lesssim a$ or $b = O(a)$ means $b\leq c a$ for some universal constant $c$; likewise, $b\gtrsim a$ or $b = \Omega(a)$ represents $b\ge ca$ for some universal constant $c$. Equivalently, $a = \Theta(b)$ means that $a$ and $b$ are of the same order up to universal constant factors.

\section{Quantum System Characterization}
\label{sec:quantum system}

\subsection{Quantum States under State Preparation Noise}
\label{sec:quantum_states}

The state of a quantum system is described by a density operator, which encodes the complete statistical properties of the system. In general, a quantum state may be described as a statistical ensemble $\{(\alpha_i, \vpsi_i)\}$, where each pure state $\vpsi_i \in \mathbb{C}^{d^n}$ occurs with probability $\alpha_i$ and satisfies $\|\vpsi_i\|_2 = 1$. The corresponding density operator is given by
\begin{eqnarray}
\label{mixed state density matrix}
\vrho = \sum_i \alpha_i \vpsi_i \vpsi_i^\dagger \in \C^{d^n \times d^n},
\end{eqnarray}
with $\sum_i \alpha_i = 1$. By construction, every density operator is positive semidefinite and has unit trace. Accordingly, we define the set of physical quantum states as
\begin{eqnarray}
    \label{Definition of physical set rho}
    \setX_{\textup{PHY}}= \Big\{ \vrho\in\C^{d^n\times d^n}:  \ \vrho \succeq {\bm 0}, \trace(\vrho) = 1    \Big\}.
\end{eqnarray}
The special case $\rank(\vrho)=1$ corresponds to a pure state, while higher-rank operators represent genuinely mixed states.

In idealized models, the prepared quantum state is assumed to coincide with the target state. In practice, however, state preparation and evolution are affected by environmental interactions, control imperfections, and decoherence, leading to deviations from the ideal description. Such deviations are naturally modeled by quantum channels acting on density operators. A noisy quantum state can be written as
\begin{equation}
\widetilde{\vrho} = \mathcal{E}(\vrho),
\end{equation}
where $\mathcal{E}$ is a completely positive trace-preserving map. Depending on the underlying physical mechanism, quantum channels capture dephasing, amplitude damping, depolarization, and other incoherent noise processes, and constitute a fundamental source of error in quantum state characterization and tomography.

Among various channel models, the depolarizing channel \cite{gross2010quantum,farooq2022robust,ivanova2023optimal,aasen2024readout,winderl2024quantum,de2024quantum,cotler2026noisy, nielsen2010quantum,wilde2013quantum} is widely adopted as a canonical benchmark due to its simplicity and isotropic nature. Under a depolarizing channel with noise strength $\lambda \in [0,1]$, the noisy quantum state is given by
\begin{equation}
\widetilde{\vrho}
=
(1-\lambda)\vrho
+
\frac{\lambda}{d^n}\mId,
\label{eq:depolarizing_channel}
\end{equation}
where $\mId\in\R^{d^n\times d^n}$ denotes the identity operator. This model can be interpreted as replacing the ideal state with the maximally mixed state with probability $\lambda$, thereby reducing distinguishability as noise increases. Throughout this paper, we adopt the depolarizing channel as a representative model for state preparation noise, which enables a tractable analysis while preserving the essential effects of channel-induced degradation.

\subsection{Quantum Measurements under Measurement Noise}
\label{sec:quantum measurements}

To characterize an unknown quantum state, one performs measurements on a large ensemble of identically prepared copies of the system. The most general class of physically realizable measurements is described by \emph{Positive Operator-Valued Measures} (POVMs)~\cite{nielsen2002quantum}, which we formally introduce below.

\begin{defi}[POVM~\cite{nielsen2002quantum}]
A Positive Operator-Valued Measure (POVM) is a collection of PSD matrices $\{\mA_1,\ldots,\mA_K\}$ satisfying
\begin{eqnarray}
\label{The defi of POVM 1}
\sum_{k=1}^K \mA_k = \mId,
\end{eqnarray}
where $\mId$ is the identity matrix.
Each POVM element $\mA_k$ is associated with a possible outcome of a quantum measurement, and the probability $p_k$ of detecting the $k$-th outcome when measuring the density operator $\vrho$ is given by
\begin{eqnarray}
\label{The defi of POVM 2}
p_k = \innerprod{\mA_k}{\vrho} \equiv \trace(\mA_k\vrho),
\end{eqnarray}
where $\sum_{k=1}^Kp_k=1$ due to \cref{The defi of POVM 1} and the fact that $\trace(\vrho) = 1$.
\end{defi}
An important class of measurements is given by informationally complete POVMs, which are sufficient to uniquely determine any quantum state. However, many practically relevant measurement schemes, such as rank-one projective measurements, are not informationally complete on their own. As a result, a single measurement setting is generally insufficient to fully reconstruct the underlying state. To address this limitation, it is standard to employ multiple distinct POVMs, thereby aggregating complementary information from different measurement settings. Specifically, in this paper we consider a collection of $Q$ POVMs $\{\mA_{q,1},\dots, \mA_{q,K}  \}_{q=1}^{Q}$, where each POVM satisfies
\begin{eqnarray}
\sum_{k=1}^{K}\mA_{q,k} = \mId, \quad \forall q \in [Q].
\end{eqnarray}

In the idealized measurement model, the POVM elements $\{\mA_{q,k}\}$ are assumed to be implemented exactly as specified, and the resulting outcome probabilities are determined solely by the Born rule applied to the quantum state. In practice, however, measurement devices are also subject to systematic imperfections and environmental disturbances, leading to deviations between the ideal POVM and the actually implemented measurement operators. Such imperfections can be naturally modeled as noise acting at the level of POVM elements, in complete analogy with state preparation noise. Inspired by the depolarizing channel model for state preparation noise, we model measurement imperfections via an analogous isotropic perturbation acting directly on the POVM elements:
\begin{eqnarray}
\label{The noisy POVM with depolarizing noise}
\wt\mA_{q,k} = (1-a) \mA_{q,k} + \frac{a}{K}\mId,
\end{eqnarray}
where $a \in [0,1]$ quantifies the strength of the measurement noise. This model captures the effect of uniform mixing with the identity operator, and preserves the POVM normalization structure.

Under this noisy measurement model, we now describe the induced probability distributions. Let $\wt{\vrho}$ denote the (possibly noisy) quantum state. For each measurement setting $q \in [Q]$, we define the corresponding probability vector
\begin{eqnarray}
\label{The noisy probability vector}
\vp_q =  \begin{bmatrix} p_{q,1} \\ \vdots \\ p_{q,K}  \end{bmatrix} = \begin{bmatrix} \< \wt{\mA}_{q,1}, \wt{\vrho} \> \\ \vdots \\ \< \wt{\mA}_{q,K}, \wt{\vrho} \>  \end{bmatrix} \in \R^K.
\end{eqnarray}

By stacking all measurement settings, we obtain the global probability vector
\begin{eqnarray}
\label{The noisy probability vector all}
\vp = \begin{bmatrix} \vp_1^\top & \cdots \vp_Q^\top  \end{bmatrix}^\top \in\R^{QK}.
\end{eqnarray}

\subsection{Empirical Probabilities under Finite-Shot Statistical Noise}
\label{sec:finite shot statistical noise}

In practical quantum experiments, the ideal probability distributions induced by a given measurement setting are not directly observable. Instead, each measurement setting is implemented repeatedly on independently prepared copies of the quantum system, and only finite-sample outcomes are available. Specifically, for the $q$-th POVM, the measurement is repeated $M$ times independently, yielding outcome counts $f_{q,k}$ for each outcome $k \in [K]$. The corresponding empirical probabilities are defined as
\begin{equation}
\wh p_{q,k} = \frac{f_{q,k}}{M}, \  k \in[K],
\label{eq:empirical-prob}
\end{equation}
where the count vector $(f_{q,1},\ldots,f_{q,K})$ follows a multinomial distribution $\operatorname{Multinomial}(M,\vp_q)$~\cite{severini2005elements}, with parameter given by the underlying probability vector $\vp_q$ induced by the noisy state and noisy POVM as defined in \eqref{The noisy probability vector}.

This construction yields a finite-shot observation model in which the empirical statistics fluctuate around the underlying probability distribution due to multinomial sampling noise. In particular, the randomness in the observations arises from both the quantum measurement process and the finite number of repetitions, leading to an additional layer of statistical noise beyond state preparation and measurement imperfections.

By stacking all measurement settings, we obtain the global empirical probability vector
\begin{equation}
\wh\vp = \begin{bmatrix}
    \wh \vp_1^\top & \cdots & \wh \vp_Q^\top
\end{bmatrix}^\top\in\R^{QK},
\label{eq:map-M-POVM1}
\end{equation}
which serves as the observed data for subsequent statistical estimation and sample complexity analysis.

\section{Structured Quantum States}
\label{sec:structured quantum states}

In this section, we discuss representative low-dimensional structures in quantum states and introduce structured models that enable scalable QST. Since a generic $n$-qubit quantum state is specified by exponentially many parameters, reconstructing an arbitrary state generally requires prohibitive computational and sample resources. Scalable QST therefore relies on exploiting additional low-dimensional structure that substantially reduces the effective number of degrees of freedom. Such structures naturally arise from locality of interactions, limited entanglement, state preparation procedures, and underlying symmetries. By restricting the reconstruction problem to these structured state families, one can often achieve accurate state estimation with significantly improved efficiency.

\paragraph*{Low-rank representation} Low-rank structure is one of the most widely exploited forms of complexity reduction in QST. Many physical quantum states are pure or nearly pure, and therefore possess density operators with low entropy that can be accurately represented by low-rank matrices~\cite{kueng2017low,guctua2020fast,francca2021fast,voroninski2013quantum,haah2017sample,ji2018pseudorandom,englbrecht2020symmetries}. Such structure substantially reduces the effective number of degrees of freedom compared to a generic density operator, thereby enabling more sample-efficient state reconstruction. Formally, we define the set of rank-$r^{\textup{LR}}$ quantum states as
\begin{eqnarray}
    \label{Definition of low rank set rho}
    \setX_{\textup{LR}}= \Big\{ \vrho\in\C^{d^n\times d^n}:  \ \vrho \succeq {\bm 0}, \trace(\vrho) = 1 , \rank(\vrho) = r^{\textup{LR}}  \Big\}.
\end{eqnarray}
Given an arbitrary estimator, one may enforce both the low-rank and physical constraints by projecting onto $\setX_{\textup{LR}}$. Such a projection can be implemented through a combination of eigenvalue thresholding and simplex projection~\cite{qin2025enhancing}, which guarantees that the resulting estimate remains a valid rank-$r^{\textup{LR}}$ density matrix. However, the associated eigendecomposition becomes computationally demanding for large-scale quantum systems. To circumvent this difficulty, it is desirable to adopt a parameterization that satisfies the low-rank and physical constraints by construction, thereby eliminating the need for explicit projections during optimization. Specifically, we employ the Burer-Monteiro factorization~\cite{burer2003nonlinear,burer2005local} and parameterize any state $\vrho\in\setX_{\textup{LR}}$ as $\vrho = \mF \mF^\dagger$ where $\mF\in\setF_{\textup{LR}}$ satisfies
\begin{eqnarray}
    \label{Definition of low rank set F}
    \setF_{\textup{LR}}= \Big\{ \mF\in\C^{d^n\times r^{\textup{LR}}}:  \ \|\mF\|_F = 1   \Big\}.
\end{eqnarray}
This factorized representation $\mF \mF^\dagger$ automatically enforces positive semidefiniteness and the rank constraint, while the normalization condition $\|\mF\|_F=1$ ensures that $\trace(\vrho)=\trace(\mF^\dagger\mF)=\|\mF\|_F^2=1$. Therefore, optimization over the constrained set $\setX_{\textup{LR}}$ can be equivalently reformulated as optimization over the lower-dimensional factor space $\setF_{\textup{LR}}$.

\paragraph*{Tensor network representations}
Low-rank structure provides an effective means of reducing the complexity of QST by restricting the state space to a low-dimensional manifold. Nevertheless, the resulting representation generally remains exponentially large in the system size. For instance, a pure $n$-qudit quantum state, corresponding to a rank-one density matrix, still requires $O(d^n)$ parameters for its description. As contemporary quantum processors continue to scale toward increasingly large system sizes, with state-of-the-art devices already exceeding one hundred qubits~\cite{preskill2018quantum,arute2019quantum,chow2021ibm}, such exponential scaling presents a fundamental obstacle to scalable quantum state reconstruction.

A more powerful form of structure arises from the observation that physically realizable quantum states typically occupy only a small and highly structured subset of the exponentially large Hilbert space. In particular, locality of interactions and restricted entanglement growth often imply that the underlying state admits a compact tensor-network representation~\cite{orus2014practical,orus2019tensor}. For one-dimensional quantum systems, such representations are particularly effective due to the limited growth of entanglement. For example, ground states of local gapped Hamiltonians and states generated by finite-time evolution under local interactions frequently obey entanglement area laws, enabling efficient tensor-network descriptions~\cite{eisert2010,noh2020efficient}. Beyond one-dimensional settings, tensor-network representations have also been successfully employed to characterize highly entangled two-dimensional quantum states, including cluster states~\cite{raussendorf2001one}, the Toric Code ground state~\cite{kitaev2003fault}, the 2D resonating valence bond model~\cite{anderson1987resonating}, and 2D AKLT model~\cite{affleck1988valence,wei2011affleck}.  Rather than representing a quantum state by a single exponentially large tensor, tensor networks decompose the global state into a collection of interconnected low-order tensors linked through virtual bonds. This decomposition can dramatically reduce the number of free parameters while preserving the essential physical properties of the state. Consequently, tensor-network representations have become a cornerstone of scalable quantum many-body modeling. We begin by introducing matrix product state (MPS) and matrix product operator (MPO), which provide compact representations of quantum states and density operators in one-dimensional systems. Building upon these foundations, we then present a novel low-rank matrix product operator (LR-MPO) representation that incorporates both low-rank structure and physical constraints within a unified tensor-network framework. Finally, we extend this construction to two-dimensional systems through projected entangled-pair state (PEPS), projected entangled-pair operator (PEPO), and the corresponding low-rank PEPO (LR-PEPO) representation.

\subparagraph*{Matrix product state} Among various tensor-network representations, the MPS is arguably the most widely used ansatz for one-dimensional quantum systems. MPS provides an efficient representation of pure quantum states by decomposing the exponentially large state vector into a sequence of low-order tensors connected through auxiliary bond indices. When the bond dimensions remain moderate, the number of free parameters grows only linearly with the system size, thereby offering an exponential reduction in storage complexity compared to a generic many-body wavefunction. Let $i_1\cdots i_\nqbit$  denote the row  index\footnote{Specifically, $i_1\cdots i_n$ represents the $(i_1+\sum_{\ell=2}^n d^{\ell-1}(i_\ell-1))$-th row.}, where $i_1,\ldots,i_\nqbit\in [d]$.
The feasible set of MPS factors is defined as
\begin{eqnarray}
    \label{the set of MPS states}
    \setF_{\textup{MPS}} &\!\!\!\!=\!\!\!\!& \Big\{ \vf\in\C^{d^n\times 1}:\  \|\vf\|_F=1, \vf(i_1 \cdots i_\nqbit) = \mX_1^{i_1}  \cdots \mX_\nqbit^{i_\nqbit}, \nonumber\\
    &\!\!\!\!\!\!\!\!& \mX_\ell^{i_\ell}\in\C^{r_{\ell-1}^{\textup{MPS}}\times r_\ell^{\textup{MPS}}}, \ell\in[n], r_0^{\textup{MPS}}=r_n^{\textup{MPS}}=1  \Big\},
\end{eqnarray}
where $r_\ell^{\textup{MPS}}$ denotes the bond dimension connecting neighboring tensors. The corresponding density operator is given by $\vrho=\vf\vf^\dagger$, which is automatically PSD, normalized, and rank one. Consequently, the MPS representation offers a compact parameterization of pure quantum states with substantially reduced storage complexity.

\subparagraph*{Matrix product operator} While MPS provides an efficient representation of pure quantum states, many applications in QST involve mixed states described by density operators. The MPO representation extends the MPS framework from state vectors to operators and has become one of the most widely used tensor-network ansatze for mixed quantum states. By introducing both row and column physical indices at each site, MPOs provide a compact representation of high-dimensional density operators while retaining the favorable scalability properties of tensor networks. Formally, the set of MPOs is defined as
\begin{eqnarray}
    \label{Definition of MPO set rho}
    \setX_{\textup{MPO}} &\!\!\!\!=\!\!\!\!& \Big\{ \vrho\in\C^{d^n\times d^n}:  \ \vrho \succeq {\bm 0}, \trace(\vrho) = 1,  \vrho(i_1 \cdots i_\nqbit, j_1\cdots j_\nqbit) = \mX_1^{i_1,j_1}  \cdots \mX_\nqbit^{i_\nqbit,j_\nqbit},\nonumber\\
    &\!\!\!\!\!\!\!\!& \mX_\ell^{i_\ell,j_\ell}\in\C^{r_{\ell-1}^{\textup{MPO}}\times r_\ell^{\textup{MPO}}}, \ell\in[n], r_0^{\textup{MPO}}=r_n^{\textup{MPO}}=1   \Big\}.
\end{eqnarray}
We emphasize that enforcing physical constraints within the MPO manifold remains a nontrivial challenge in QST. In particular, the tensor-network structure imposed by prescribed bond dimensions and the simplex constraints associated with valid density operators are generally difficult to satisfy simultaneously. Following~\cite{qin2025enhancing}, one possible strategy is to employ an approximate projection procedure consisting of a tensor-train singular value decomposition (TT-SVD)~\cite{oseledets2011tensor} followed by a simplex projection~\cite{chen2011projection}. The former restores the desired MPO structure, whereas the latter enforces positive semidefiniteness and unit trace. Despite its practical effectiveness, this sequential projection scheme does not preserve both properties exactly. In particular, the simplex projection is performed on the reconstructed density matrix and may perturb the tensor-network structure obtained from the TT-SVD step. Consequently, the resulting estimate generally no longer belongs to the original MPO manifold and may violate the prescribed bond-dimension constraints $[r_1^{\textup{MPO}},\cdots,r_{n-1}^{\textup{MPO}}]$. An alternative route is provided by matrix product density operators (MDPOs)~\cite{verstraete2004matrix}, which guarantee positive semidefiniteness by construction through local completely positive factorizations. However, the resulting parameterization introduces additional structural complexity and often leads to challenging optimization problems in large-scale settings. These limitations motivate the development of more compact representations that simultaneously preserve physical feasibility and computational tractability.

\subparagraph*{Low-rank matrix product operator}
Despite their expressive power, MPO representations do not readily enforce the physical constraints required of quantum states. As discussed above, maintaining the prescribed MPO structure while simultaneously guaranteeing positive semidefiniteness and unit trace remains challenging in practice. This motivates the search for alternative parameterizations that preserve the favorable scalability of tensor networks while enforcing physical feasibility by construction. A notable feature of many physical quantum states is the coexistence of low-rank structure and tensor-network locality. On the one hand, low-temperature thermal states and ground states of local Hamiltonians are often well approximated by low-rank density operators. On the other hand, the same states typically admit efficient tensor-network representations due to locality and restricted entanglement. These complementary properties suggest combining low-rank factorization with tensor-network parameterization within a unified framework.

Motivated by this observation, we introduce the LR-MPO representation. Specifically, we parameterize the density operator as $\vrho=\mF\mF^\dagger $ where the factor matrix $\mF\in\setF_{\textup{LR-MPO}}$ . Formally, we define the feasible set
\begin{eqnarray}
    \label{the set of LR-MPO states}
    \setF_{\textup{LR-MPO}} &\!\!\!\!=\!\!\!\!& \Big\{ \mF\in\C^{d^n\times r^{\textup{LR}}}:\ \|\mF\|_F=1, \mF(i_1 \cdots i_\nqbit,j) = \mX_1^{i_1} \cdots \mX_{B}^{i_{B},j} \cdots \mX_\nqbit^{i_\nqbit}, j\in[r^{\textup{LR}}], \nonumber\\
    &\!\!\!\!\!\!\!\!&  \mX_\ell^{i_\ell}\in\C^{r_{\ell-1}^{\textup{LR-MPO}}\times r_\ell^{\textup{LR-MPO}}}, \ell\in[n]\setminus\{B\} \  \textup{and} \  \mX_B^{i_B,j}\in\C^{r_{B-1}^{\textup{LR-MPO}}\times r_B^{\textup{LR-MPO}}}, \ r_0^{\textup{LR-MPO}}=r_n^{\textup{LR-MPO}}=1 \Big\}.
\end{eqnarray}
Relative to a standard MPS, the only modification is the introduction of the rank index $j$ at a designated site $B$. The location of $B$ is arbitrary and does not alter the resulting model class. For larger values of $r^{\textup{LR}}$, the index $j$ may be further decomposed into multiple auxiliary indices, leading to equivalent block-structured tensor-network representations; however, a single-index formulation is adopted here for notational simplicity.

The LR-MPO representation inherits several desirable properties from both constituent structures. First, the factorization $\vrho=\mF\mF^\dagger$ guarantees positive semidefiniteness by construction, while the normalization constraint $\|\mF\|_F=1$ ensures unit trace. Second, the MPO structure imposed on $\mF$ preserves the locality-induced compression characteristic of tensor networks. Consequently, LR-MPO simultaneously exploits low-rank spectral structure and tensor-network locality, providing a compact and physically consistent representation of quantum states. In addition, we note that the distinction between LR-MPO and conventional MPDO lies in the underlying parameterization. MPDO represents density operators through local completely positive tensor factorizations, whereas LR-MPO is built upon a global low-rank decomposition. In particular, each column of $\mF$ admits an MPS representation, so that $\vrho=\mF\mF^\dagger$ may be viewed as a collection of correlated MPS components coupled through the low-rank factorization. This perspective establishes a direct connection between low-rank matrix recovery and tensor-network representations, while retaining the computational advantages of both frameworks.

\subparagraph*{Projected entangled-pair state} Many quantum systems of practical interest are defined on two-dimensional lattices, where the entanglement structure is generally more intricate than that encountered in one-dimensional settings. In such cases, an MPS representation with moderate bond dimension may no longer provide an efficient description of the underlying quantum state. A natural extension of MPS to higher-dimensional systems is the PEPS representation, which has become a standard framework for describing two-dimensional quantum many-body states. By arranging local tensors according to the geometry of the underlying lattice and connecting neighboring tensors through virtual bonds, PEPS can efficiently represent a broad class of quantum states exhibiting area-law entanglement. Consider an $n$-qudit quantum system arranged on a two-dimensional $q\times p$ lattice, where $n=qp$. Let $\vf\in\C^{d^n\times 1}$ denote a pure state and $\vrho=\vf\vf^\dagger$ be the corresponding density operator. The PEPS representation parameterizes $\vf$ through a network of local tensors defined on the underlying lattice. Formally, the feasible set of PEPS factors is given by
\begin{eqnarray}
    \label{the set of PEPS states}
    \setF_{\textup{PEPS}} &\!\!\!\!=\!\!\!\!& \Big\{ \vf\in\C^{d^n\times 1}:\  \|\vf\|_F=1, \vf(i_{\ul{1}\,\ul{1}}\cdots i_{\ul{q}\,\ul{p}}) = \hspace{-1cm} \sum_{\substack{
      s_{\ul{a}\,\ul{b-1},\,\ul{a}\,\ul{b}}\;
      s_{\ul{a-1}\,\ul{b},\,\ul{a}\,\ul{b}} \\
      s_{\ul{a}\,\ul{b},\,\ul{a}\,\ul{b+1}}\;
      s_{\ul{a}\,\ul{b},\,\ul{a+1}\,\ul{b}} \\
      a\in[q],\, b\in[p]}}  \hspace{-1cm}
   \;\Pi_{a=1}^{q} \Pi_{b=1}^{p}
   \mX_{\ul{a}\,\ul{b}}^{\,i_{\ul{a}\,\ul{b}} }
   (s_{\ul{a}\,\ul{b-1},\,\ul{a}\,\ul{b}},\,
    s_{\ul{a-1}\,\ul{b},\,\ul{a}\,\ul{b}},\, s_{\ul{a}\,\ul{b},\,\ul{a}\,\ul{b+1}},\, \nonumber\\
    &\!\!\!\!\!\!\!\!& s_{\ul{a}\,\ul{b},\,\ul{a+1}\,\ul{b}}),
    \mX_{\ul{a}\,\ul{b}}^{\,i_{\ul{a}\,\ul{b}} }\in\C^{r_{\ul{a}\,\ul{b-1},\,\ul{a}\,\ul{b}}^{\textup{PEPS}}\times r_{\ul{a-1}\,\ul{b},\,\ul{a}\,\ul{b}}^{\textup{PEPS}}
\times r_{\ul{a}\,\ul{b},\,\ul{a}\,\ul{b+1}}^{\textup{PEPS}}\times r_{\ul{a}\,\ul{b},\,\ul{a+1}\,\ul{b}}^{\textup{PEPS}}}, a\in[q], b\in[p], r_{\ul{a}\,\ul{0},\,\ul{a}\,\ul{1}}^{\textup{PEPS}} =
 r_{\ul{0}\,\ul{b},\,\ul{1}\,\ul{b}}^{\textup{PEPS}} = \nonumber\\
&\!\!\!\!\!\!\!\!&  r_{\ul{a}\,\ul{p},\,\ul{a}\,\ul{p+1}}^{\textup{PEPS}} =
 r_{\ul{q}\,\ul{b},\,\ul{q+1}\,\ul{b}}^{\textup{PEPS}} = 1  \Big\}.
\end{eqnarray}

\subparagraph*{Projected entangled-pair operator} Similar to the relationship between MPS and MPO in one-dimensional systems, the PEPS representation is primarily designed for pure quantum states and therefore cannot directly capture general mixed states. In many applications of QST, however, the target state is described by a density operator rather than a wavefunction, necessitating a tensor-network representation in operator space. A natural extension of PEPS is provided by the PEPO representation, which generalizes tensor-network parameterizations from state vectors to density operators. By introducing both row and column physical indices at each lattice site, PEPO enables a compact representation of mixed quantum states while preserving the locality structure inherent in the underlying two-dimensional lattice. Consider an $n$-qudit quantum system arranged on a $q\times p$ lattice. Analogous to the MPO representation in one-dimensional systems, PEPO expresses the density operator as a contraction of local tensors associated with the lattice sites. The corresponding feasible set is defined as
\begin{eqnarray}
    \label{Definition of PEPO set rho}
    \setX_{\textup{PEPO}} &\!\!\!\!=\!\!\!\!& \Big\{ \vrho\in\C^{d^n\times d^n}:  \ \vrho \succeq {\bm 0}, \trace(\vrho) = 1,  \vrho(i_{\ul{1}\,\ul{1}}\cdots i_{\ul{q}\,\ul{p}}, \; j_{\ul{1}\,\ul{1}}\cdots j_{\ul{q}\,\ul{p}}) = \hspace{-1cm} \sum_{\substack{
      s_{\ul{a}\,\ul{b-1},\,\ul{a}\,\ul{b}}\;
      s_{\ul{a-1}\,\ul{b},\,\ul{a}\,\ul{b}} \\
      s_{\ul{a}\,\ul{b},\,\ul{a}\,\ul{b+1}}\;
      s_{\ul{a}\,\ul{b},\,\ul{a+1}\,\ul{b}} \\
      a\in[q],\, b\in[p]}} \hspace{-1cm}
   \;\Pi_{a=1}^{q} \Pi_{b=1}^{p}
   \mX_{\ul{a}\,\ul{b}}^{\,i_{\ul{a}\,\ul{b}}, j_{\ul{a}\,\ul{b}}}
   (s_{\ul{a}\,\ul{b-1},\,\ul{a}\,\ul{b}},\, \nonumber\\
   &\!\!\!\!\!\!\!\!& s_{\ul{a-1}\,\ul{b},\,\ul{a}\,\ul{b}},\, s_{\ul{a}\,\ul{b},\,\ul{a}\,\ul{b+1}},\,
    s_{\ul{a}\,\ul{b},\,\ul{a+1}\,\ul{b}}), \mX_{\ul{a}\,\ul{b}}^{\,i_{\ul{a}\,\ul{b}},\,j_{\ul{a}\,\ul{b}}}\in\C^{r_{\ul{a}\,\ul{b-1},\,\ul{a}\,\ul{b}}^{\textup{PEPO}}\times r_{\ul{a-1}\,\ul{b},\,\ul{a}\,\ul{b}}^{\textup{PEPO}}
\times r_{\ul{a}\,\ul{b},\,\ul{a}\,\ul{b+1}}^{\textup{PEPO}}\times r_{\ul{a}\,\ul{b},\,\ul{a+1}\,\ul{b}}^{\textup{PEPO}}}, a\in[q], b\in[p],\nonumber\\
    &\!\!\!\!\!\!\!\!&  r_{\ul{a}\,\ul{0},\,\ul{a}\,\ul{1}}^{\textup{PEPO}} =
 r_{\ul{0}\,\ul{b},\,\ul{1}\,\ul{b}}^{\textup{PEPO}} =
 r_{\ul{a}\,\ul{p},\,\ul{a}\,\ul{p+1}}^{\textup{PEPO}} =
 r_{\ul{q}\,\ul{b},\,\ul{q+1}\,\ul{b}}^{\textup{PEPO}} = 1   \Big\}.
\end{eqnarray}
Unlike MPO representations, which admit practical approximate projection procedures for enforcing both tensor-network structure and physical constraints, no analogous approach is currently available for PEPOs. Consequently, obtaining physically feasible PEPO representations remains considerably more challenging. This motivates the development of alternative tensor-network parameterizations that preserve the expressive power of PEPOs while enforcing physical feasibility by construction.

\subparagraph*{Low-rank projected entangled-pair operator} Following the same principle underlying the LR-MPO construction, we introduce a LR-PEPO representation for two-dimensional quantum systems. Specifically, we parameterize the density operator as $\vrho=\mF\mF^\dagger$ where the factor matrix $\mF\in\setF_{\textup{LR-PEPO}}$ . Formally, we define the feasible set
\begin{eqnarray}
    \label{the set of LR-PEPO states}
    \setF_{\textup{LR-PEPO}} &\!\!\!\!=\!\!\!\!& \Big\{ \mF\in\C^{d^n\times r^{\textup{LR}}}:\ \|\mF\|_F=1, \mF(i_{\ul{1}\,\ul{1}}\cdots i_{\ul{q}\,\ul{p}},j) = \hspace{-1cm} \sum_{\substack{
      s_{\ul{a}\,\ul{b-1},\,\ul{a}\,\ul{b}}\;
      s_{\ul{a-1}\,\ul{b},\,\ul{a}\,\ul{b}} \\
      s_{\ul{a}\,\ul{b},\,\ul{a}\,\ul{b+1}}\;
      s_{\ul{a}\,\ul{b},\,\ul{a+1}\,\ul{b}} \\
      a\in[q],\, b\in[p], \ul{a}\,\ul{b} \neq \ul{B_1}\,\ul{B_2}}}\hspace{-1cm}
   \;\Pi_{a=1}^{q} \Pi_{b=1}^{p}
   \mX_{\ul{a}\,\ul{b}}^{\,i_{\ul{a}\,\ul{b}} }
   (s_{\ul{a}\,\ul{b-1},\,\ul{a}\,\ul{b}},\,
    s_{\ul{a-1}\,\ul{b},\,\ul{a}\,\ul{b}},\, s_{\ul{a}\,\ul{b},\,\ul{a}\,\ul{b+1}},\,  \nonumber\\
    &\!\!\!\!\!\!\!\!& s_{\ul{a}\,\ul{b},\,\ul{a+1}\,\ul{b}})\mX_{\ul{B_1}\,\ul{B_2}}^{\,i_{\ul{B_1}\,\ul{B_2}}, j }
   (s_{\ul{B_1}\,\ul{B_2-1},\,\ul{B_1}\,\ul{B_2}},\,
    s_{\ul{B_1-1}\,\ul{B_2},\,\ul{B_1}\,\ul{B_2}},\, s_{\ul{B_1}\,\ul{B_2},\,\ul{B_1}\,\ul{B_2+1}},\,  s_{\ul{B_1}\,\ul{B_2},\,\ul{B_1+1}\,\ul{B_2}}), j\in[r^{\textup{LR}}], \nonumber\\
    &\!\!\!\!\!\!\!\!&  \mX_{\ul{a}\,\ul{b}}^{\,i_{\ul{a}\,\ul{b}} }\in\C^{r_{\ul{a}\,\ul{b-1},\,\ul{a}\,\ul{b}}^{\textup{LR-PEPO}}\times r_{\ul{a-1}\,\ul{b},\,\ul{a}\,\ul{b}}^{\textup{LR-PEPO}}
\times r_{\ul{a}\,\ul{b},\,\ul{a}\,\ul{b+1}}^{\textup{LR-PEPO}}\times r_{\ul{a}\,\ul{b},\,\ul{a+1}\,\ul{b}}^{\textup{LR-PEPO}}}, a\in[q]\setminus \{ B_1\}, b\in[p]\setminus \{ B_2\},\nonumber\\
&\!\!\!\!\!\!\!\!&\mX_{\ul{B_1}\,\ul{B_2}}^{\,i_{\ul{B_1}\,\ul{B_2}}, j }\in\C^{r_{\ul{B_1}\,\ul{B_2-1},\,\ul{B_1}\,\ul{B_2}}^{\textup{LR-PEPO}}\times r_{\ul{B_1-1}\,\ul{B_2},\,\ul{B_1}\,\ul{B_2}}^{\textup{LR-PEPO}}
\times r_{\ul{B_1}\,\ul{B_2},\,\ul{B_1}\,\ul{B_2+1}}^{\textup{LR-PEPO}}\times r_{\ul{B_1}\,\ul{B_2},\,\ul{B_1+1}\,\ul{B_2}}^{\textup{LR-PEPO}}},\nonumber\\
    &\!\!\!\!\!\!\!\!&
 r_{\ul{a}\,\ul{0},\,\ul{a}\,\ul{1}}^{\textup{LR-PEPO}} =
 r_{\ul{0}\,\ul{b},\,\ul{1}\,\ul{b}}^{\textup{LR-PEPO}} =
 r_{\ul{a}\,\ul{p},\,\ul{a}\,\ul{p+1}}^{\textup{LR-PEPO}} =
 r_{\ul{q}\,\ul{b},\,\ul{q+1}\,\ul{b}}^{\textup{LR-PEPO}} = 1   \Big\}.
\end{eqnarray}
In this construction, the auxiliary index $j$ is incorporated into the local tensor at site $(B_1,B_2)$, while all other tensors are shared across different columns of $\mF$. The low-rank structure is introduced through a single auxiliary index embedded in one local tensor, which induces a global coupling across the entire tensor-network contraction. As a result, the proposed LR-PEPO representation simultaneously preserves the two-dimensional tensor-network structure of PEPOs and guarantees physical feasibility.

\paragraph*{Sparse representations} Sparse structures arise when a quantum state admits a compact description in a suitable basis, such that only a small number of coefficients are significantly nonzero. Although sparsity has received comparatively less attention than low-rank or tensor-network structures in QST, it naturally appears in several physical quantum states. Representative examples include generalized GHZ states~\cite{jameson2024optimal}, and Dicke states with a small number of excitations~\cite{bartschi2019deterministic}, all of which have support on only a small subset of computational-basis vectors. More broadly, sparsity provides an alternative notion of complexity that can be leveraged when the target state is concentrated on a small subset of basis elements or basis functions. Formally, we define the set of $s$-sparse quantum states as
\begin{eqnarray}
    \label{Definition of sparse set rho}
    \setX_{\textup{S}}= \Big\{ \vrho\in\C^{d^n\times d^n}:  \ \|\vrho\|_{0} = s, \vrho \succeq {\bm 0}, \trace(\vrho) = 1   \Big\},
\end{eqnarray}
where $\|\vrho\|_{0}$ denotes the number of nonzero entries of $\vrho$. However, directly enforcing sparsity on the density matrix while simultaneously satisfying the physical constraints of positive semidefiniteness and unit trace can be challenging. Moreover, as illustrated by the examples above, many physical sparse quantum states also exhibit low-rank structure. Motivated by these considerations, one may instead consider a Burer-Monteiro factorization $\vrho=\mF\mF^\dagger$ and impose sparsity on the factor matrix $\mF$. Specifically, we define the corresponding low-rank and sparse (LR-S) factor set as
\begin{eqnarray}
    \label{Definition of sparse set F}
    \setF_{\textup{LR-S}}= \Big\{ \mF\in\C^{d^n\times r^{\textup{LR}}}:  \ \|\mF\|_{0} = t, \|\mF\|_F = 1   \Big\}.
\end{eqnarray}
Under this parameterization, positive semidefiniteness is guaranteed automatically by construction, while sparsity and low-rankness are simultaneously encoded through the factor matrix $\mF$. Such a formulation naturally captures structured quantum states that admit sparse and low-rank descriptions, including GHZ, W, and Dicke states. Furthermore, if $\mF\in\setF_{\textup{LR-S}}$, then the corresponding density matrix satisfies
\begin{eqnarray}
    \label{Upper bound of density matrix from the factor}
    \|\vrho \|_0\leq t^2,
\end{eqnarray}
which shows that sparsity in the factor matrix induces sparsity in the resulting density matrix.

\section{Theoretical Guarantees for Structured Quantum State Tomography}
\label{sec:structured quantum state tomography}

\subsection{Characterization of Measurement Quality}
\label{subsec:Measurement Quality Characterization}

Throughout this section, we use $\setX$ to denote the admissible family of structured quantum states under consideration, which can be specified either directly through constraints on the density matrix $\vrho$ or indirectly through a factorized parameterization $\vrho=\mF \mF^\dagger$, where $\mF \in\setF$. For a given state family $\setX$, the reconstruction accuracy of structured QST depends critically on the information provided by the underlying measurement ensemble. Rather than restricting our analysis to a particular class of POVMs, we introduce a set of abstract quantities that characterize the quality of the measurements over the admissible state family $\setX$. These quantities capture both the distinguishability between different quantum states and the higher-order interactions arising in the statistical analysis of the proposed estimators. As a result, the subsequent sample complexity bounds can be expressed in a unified form and applied to a broad range of measurement schemes, including deterministic informationally complete POVMs \cite{matthews2009distinguishability,kueng2017low} as well as random measurement ensembles such as Haar random projective measurements \cite{haah2017sample,qin2024quantum,qin2025enhancing} and unitary $t$-designs \cite{gross2007evenly,dankert2009exact,roy2009unitary}.

To quantify the measurement quality, we define the following complexity parameters.
\begin{eqnarray}
    \label{lower constant of second order information}
    &\!\!\!\!\!\!\!\!&\alpha_1(Q,K) := \inf_{\vrho, \vrho^\star \in \setX}\frac{\sum_{q=1}^{Q}\sum_{k=1}^{K}\<\mA_{q,k}, \vrho - \vrho^\star\>^2}{\|\vrho - \vrho^\star\|_F^2},\\
    \label{upper constant of second order information}
    &\!\!\!\!\!\!\!\!&\alpha_2(Q,K) := \sup_{\vrho, \vrho^\star \in \setX}\frac{\sum_{q=1}^{Q}\sum_{k=1}^{K}\<\mA_{q,k}, \vrho - \vrho^\star\>^2}{\|\vrho - \vrho^\star\|_F^2},\\
    \label{upper constant of third order information}
    &\!\!\!\!\!\!\!\!&\beta(Q,K) := \sup_{\vrho, \vrho^\star \in \setX}\frac{\sum_{q=1}^{Q}\sum_{k=1}^{K}\<\mA_{q,k}, \vrho - \vrho^\star\>^2\<\mA_{q,k}, \vrho^\star  \>}{\|\vrho - \vrho^\star\|_F^2},
\end{eqnarray}
where $\setX\subset\setX_{\textup{PHY}}$ denotes the admissible set of structured quantum states.  These parameters characterize different aspects of the measurement ensemble with respect to the admissible state family $\setX$.  $(i)$ The lower complexity parameter, $\alpha_1(Q,K)$, measures the minimum distinguishability between different states after projection onto the measurement operators. Specifically, it provides a uniform lower bound on the measurement-induced quadratic form over all admissible perturbations $\vrho-\vrho^\star$. A larger value of $\alpha_1(Q,K)$ indicates that distinct states produce more separated measurement statistics, thereby improving identifiability. $(ii)$ The upper complexity parameter $\alpha_2(Q,K)$ characterizes the largest measurement-induced quadratic form over the admissible perturbation set. Together, $\alpha_1(Q,K)$ and $\alpha_2(Q,K)$ quantify the conditioning of the measurement operator restricted to $\setX$, describing how uniformly the measurement ensemble preserves the geometry of admissible state differences. In particular, when $\alpha_1(Q,K)$ and $\alpha_2(Q,K)$ are of the same order, the measurement ensemble behaves nearly isotropically over $\setX$. $(iii)$ The parameter, $\beta(Q,K)$, captures a weighted second-order moment of the measurement ensemble. Unlike $\alpha_1(Q,K)$ and $\alpha_2(Q,K)$, which depends only on the measurement geometry, $\beta(Q,K)$ also incorporates the outcome probabilities $\<\mA_{q,k},\vrho^\star\>$ induced by the underlying state. Consequently, $\beta(Q,K)$ quantifies the average measurement energy of admissible perturbations under the induced measurement distribution.

Although the three parameters characterize different aspects of the measurement ensemble, they enter the recovery analysis only through the normalized complexity ratio
\begin{eqnarray}
    \label{ratio of SNR}
    \frac{(1-a)\beta(Q,K) + a\alpha_2(Q,K)/K}{(\alpha_1(Q,K))^2},
\end{eqnarray}
which serves as the fundamental complexity parameter governing the statistical efficiency of the measurement ensemble. Specifically, the numerator captures the magnitude of statistical fluctuations arising from both the probability-weighted second-order moment and the worst-case measurement energy, whereas the denominator quantifies the intrinsic distinguishability of admissible states. Consequently, a smaller ratio corresponds to a more informative and better-conditioned measurement ensemble, leading to lower sample complexity and more accurate state reconstruction. As shown below, this quantity directly determines the recovery guarantees established in this paper.

The complexity parameters $\alpha_1(Q,K)$, $\alpha_2(Q,K)$, and $\beta(Q,K)$ provide an abstract characterization of the quality of a general POVM ensemble. To illustrate their behavior in concrete settings, we next evaluate these quantities for several measurement schemes commonly used in QST, including spherical $3$-designs and $\delta$-approximate spherical $3$-designs (corresponding to the case $Q=1$), unitary $3$-designs, and Haar random projective measurements. According to \cite[Eqs.~(7) and (47)]{qin2024sample} and \cite[Eqs.~(S35) and (S36)]{huang2020predicting}, for $\setX\subset\setX_{\textup{PHY}}$, we obtain
\begin{eqnarray}
\label{parameters of different POVMs}
\frac{(1-a)\beta(Q,K) + a\alpha_2(Q,K)/K}{(\alpha_1(Q,K))^2} = \begin{dcases}
    \Theta(1), &  \text{Spherical $3$-designs}, \\
    \Theta\bigg(\frac{1+\delta}{(1-\delta)^2}\bigg), &  \text{$\delta$-approximate spherical $3$-designs}, \\
    \Theta\bigg(\frac{1}{Q}\bigg), &  \text{Haar/Unitary $3$-designs}.
\end{dcases}
\end{eqnarray}
For spherical $3$-designs, the scaling in \eqref{parameters of different POVMs} follows directly from deterministic evaluations of the quantities defined in \eqref{lower constant of second order information} and \eqref{upper constant of third order information}. The same derivation extends to $\delta$-approximate spherical $3$-designs, where the approximation error is reflected in the additional factor $(1+\delta)/(1-\delta)^2$. For unitary $3$-designs and Haar random projective measurements, the corresponding scaling is obtained by evaluating the measurement moments appearing in \eqref{lower constant of second order information}, \eqref{upper constant of second order information} and \eqref{upper constant of third order information} with respect to the underlying measurement distribution, i.e., by taking expectations over the random measurement operators. Notably, compared with (approximate) spherical $3$-designs, unitary $3$-designs and Haar random projective measurements exhibit an additional factor of $1/Q$ improvement in the effective complexity parameter. Since the ratio directly governs the sample complexity bounds derived below, this scaling implies that the statistical complexity of the QST decreases with the number of measurement settings. Consequently, randomized measurement ensembles require proportionally fewer state copies to achieve the same reconstruction accuracy as $Q$ increases.

\subsection{Sample Complexity}
\label{subsec:sample omplexity}

A central difficulty in analyzing structured QST lies in establishing uniform control of the empirical loss over the admissible state class $\setX$. To this end, $\epsilon$-net and covering number theory provide a powerful framework for quantifying the complexity of structured state classes and deriving finite-sample guarantees. However, constructing an $\epsilon$-cover for $\setX$ is challenging because $\setX$ is characterized by nonconvex structural constraints, such as sparsity, low-rankness and tensor network representations. To address this issue, we exploit the factorized structure of the hypothesis class. Specifically, every admissible state $\vrho \in \setX$ is characterized by $L$ local components, where the $i$-th component belongs to a constraint set $\setX_i$, $i\in[L]$,  together with an additional global constraint set $\setX_{L+1}$. Here $L$ is model dependent and represents the number of constituent components in the chosen parameterization. For instance, $L=1$ for sparse or low-rank factorizations, whereas $L=n$ for tensor network representations. We define the factorized class $\widetilde{\setX} = \{\vrho : \text{the local components of } \vrho \text{ belong to } \setX_1 \times \cdots \times \setX_L\}$, so that $\setX = \widetilde{\setX} \cap \setX_{L+1}$. Let $N_\epsilon(\setX)$ denote the covering number of $\setX$ under the Frobenius norm, i.e., the minimum number of Frobenius-norm balls of radius $\epsilon$ required to cover $\setX$. By the monotonicity of covering numbers under set inclusion with the Cartesian-product structure of the parameter space, we obtain the following bound on the covering number: $N_{\epsilon}(\setX) = N_{\epsilon}(\widetilde{\setX} \cap \setX_{L+1}) \leq \min\{N_{\epsilon}(\widetilde{\setX}),  N_{\epsilon}(\setX_{L+1}) \}\leq \min\{\Pi_{i=1}^{L} N_{\epsilon}(\setX_i),  N_{\epsilon}(\setX_{L+1}) \}$. We summarize the resulting covering number bounds for the considered state classes in the following lemma.
\begin{lemma}[Covering numbers of structured quantum-state classes]
\label{conclusion of all covering numbers}
Let $N_{\textup{type}}$ denote the covering number of the parameter set associated with a structured class of quantum states under the Frobenius norm, where
$\textup{type}$ indexes the class under consideration; the precise definitions of these parameter sets are given in Appendix~\ref{Proof of ocvering number all}. The logarithms of the covering numbers for the structured quantum-state classes considered in this paper are summarized in Table~\ref{tab:covering-numbers}.
\begin{table*}[!ht]
\centering
\footnotesize
\caption{Log-covering numbers $\log N_{\textup{type}}$ for structured quantum-state classes.}
\label{tab:covering-numbers}
\renewcommand{\arraystretch}{1.5}
\begin{tabular}{cl}
\hline
\textbf{Class} & \textbf{$\log N_{\textup{type}}$} \\
\hline
$\setX_{\textup{PHY}}$ &
$O(d^{2n})$ \\
$\setX_{\textup{S}}$ &
$O\big(s\log\tfrac{d^{2n}}{s}\big)$ \\
$\setX_{\textup{LR}}/\setF_{\textup{LR}}$ &
$O(d^n r^{\textup{LR}})$ \\
$\setF_{\textup{LR-S}}$ &
$O\big(t^2\log\tfrac{d^{2n}}{t^2}\big)$ \\
$\setF_{\textup{MPS}}$ &
$O\big(\sum_{\ell=1}^{n} d\,r_{\ell-1}^{\textup{MPS}} r_\ell^{\textup{MPS}} \log n\big)$ \\
$\setX_{\textup{MPO}}$ &
$O\big(\sum_{\ell=1}^{n} d^2\,r_{\ell-1}^{\textup{MPO}} r_\ell^{\textup{MPO}} \log n\big)$ \\
$\setF_{\textup{LR-MPO}}$ &
$O\big(\sum_{\ell=1}^{n} d^2\,(r_{\ell-1}^{\textup{LR\text{-}MPO}})^2 (r_\ell^{\textup{LR\text{-}MPO}})^2 \log n\big)$ \\
$\setF_{\textup{PEPS}}$ &
$O\big(\sum_{a=1}^{q}\sum_{b=1}^{p} d\,
r^{\textup{PEPS}}_{a,b-1} r^{\textup{PEPS}}_{a-1,b}
r^{\textup{PEPS}}_{a,b+1} r^{\textup{PEPS}}_{a+1,b} \log n\big)$ \\
$\setX_{\textup{PEPO}}$ &
$O\big(\sum_{a=1}^{q}\sum_{b=1}^{p} d^2\,
r^{\textup{PEPO}}_{a,b-1} r^{\textup{PEPO}}_{a-1,b}
r^{\textup{PEPO}}_{a,b+1} r^{\textup{PEPO}}_{a+1,b} \log n\big)$ \\
$\setF_{\textup{LR-PEPO}}$ &
$O\big(\sum_{a=1}^{q}\sum_{b=1}^{p} d^2\,
(r^{\textup{LR\text{-}PEPO}}_{a,b-1})^2
(r^{\textup{LR\text{-}PEPO}}_{a-1,b})^2
(r^{\textup{LR\text{-}PEPO}}_{a,b+1})^2
(r^{\textup{LR\text{-}PEPO}}_{a+1,b})^2 \log n\big)$ \\
\hline
\end{tabular}
\end{table*}
\end{lemma}
The detailed proof is provided in Appendix~\ref{Proof of ocvering number all}. The bounds in \Cref{conclusion of all covering numbers} show that the complexity of each structured quantum-state class is characterized by the logarithm of its covering number, which depends explicitly on the underlying structural parameters, such as sparsity, rank, or tensor-network bond dimensions. Although the covering-number scaling differs across models, the subsequent analyses depend only on these covering-number estimates. Consequently, \Cref{conclusion of all covering numbers} provides a unified complexity characterization that serves as the foundation for the sample complexity analysis developed in the remainder of this section.

We begin by formulating the estimators used for QST, which will serve as the basis for the subsequent sample complexity analysis. The estimators are constructed based on the empirical probability vector $\widehat{\vp}$ defined in Sections~\ref{sec:quantum_states}--\ref{sec:finite shot statistical noise}, which is generated under the combined effects of state preparation noise, measurement noise, and finite-shot statistical fluctuations. To reconstruct the unknown quantum state, one must specify a forward measurement model relating the density operator to the observed empirical probabilities. Depending on whether the experimental noise characteristics are available during reconstruction, we consider the following two estimation settings.

\paragraph*{Noise-Aware Quantum State Tomography}

We first consider the idealized setting in which the state preparation and measurement noise models, together with their corresponding noise parameters $(\lambda,a)$, are assumed to be known through prior device calibration or independent system characterization. Consequently, the reconstruction employs the same noisy forward model that governs the observation process. The corresponding least-squares objective is given by
\begin{eqnarray}
\label{loss function of unbiased solution}
f(\vrho)  = \frac{1}{2QK}\sum_{q=1}^{Q}\sum_{k=1}^{K} \bigg( \bigg\<(1-a) \mA_{q,k} + \frac{a}{K}\mId,   (1-\lambda)\vrho + \frac{\lambda}{d^n}\mId  \bigg\> -  \wh p_{q,k}   \bigg)^2.
\end{eqnarray}
Since the forward model coincides with the physical data-generation mechanism, this formulation is statistically well specified. The corresponding estimator is defined as
\begin{eqnarray}
\label{minimizer loss function of unbiased solution}
\wh\vrho = \argmin_{\vrho\in\setX_{\textup{type}}}f(\vrho).
\end{eqnarray}

Since many structured quantum-state models admit low-rank or tensor-network factorizations, it is often more convenient to optimize over an auxiliary parameterization rather than directly over the density operator. In particular, we consider a factorization of the form $\vrho = \mF \mF^\dagger$ with $\mF \in {\setF}_{\textup{type}}$. Substituting this parameterization into \eqref{loss function of unbiased solution} yields the equivalent optimization problem
\begin{eqnarray}
\label{loss function factorized}
f(\mF) = \frac{1}{2QK}\sum_{q=1}^{Q}\sum_{k=1}^{K} \bigg( \bigg\< (1-a)\mA_{q,k} + \frac{a}{K}\mId,\, (1-\lambda)\mF\mF^\dagger + \frac{\lambda}{d^n}\mId \bigg\> - \wh p_{q,k} \bigg)^2.
\end{eqnarray}
The corresponding estimator in the factorized parameter space is given by
\begin{eqnarray}
\label{minimizer F}
\wh \mF = \argmin_{\mF\in {\setF}_{\textup{type}}} f(\mF), \qquad \wh\vrho = \wh\mF\wh\mF^\dagger.
\end{eqnarray}
Both formulations are equivalent in terms of the induced density operator, while the factorized representation is more amenable to low-dimensional structured optimization and plays a central role in the subsequent sample complexity analysis. The following theorem characterizes its reconstruction error in terms of the intrinsic complexity of the underlying state class, the measurement ensemble, and the state preparation and measurement noise levels.
\begin{theorem}
\label{label:sample complexity of unbiased solution}

Suppose $Q$ POVMs $\{\mA_{q,1},\cdots, \mA_{q,K}\}_{q\in[Q]}$ satisfy \eqref{lower constant of second order information}, \eqref{upper constant of second order information}, and \eqref{upper constant of third order information}. Using each POVM to measure the quantum state $M$ times yields the empirical probabilities $\{\wh p_{q,k}\}_{q\in[Q],k\in[K]}$. Then, with probability at least $1-e^{-\Omega(\log N_{\textup{type}})}$, the estimator $\wh\vrho$ obtained from either the constrained least-squares formulation \eqref{loss function of unbiased solution} or its factorized counterpart \eqref{loss function factorized} satisfies
\begin{eqnarray}
\label{upper bound of solution for prior knowledge summary}
\|\wh\vrho - \vrho^\star\|_1 \leq O\bigg(\sqrt{\frac{((1-a)\beta(Q,K) + a\alpha_2(Q,K)/K)\rank(\vrho^\star)(\log N_{\textup{type}})}{(\alpha_1(Q,K))^2 (1-a)^2(1-\lambda)^2M}}\bigg),
\end{eqnarray}
where $N_{\textup{type}}$ denotes the covering number associated with the underlying structured quantum-state class summarized in Table~\ref{tab:covering-numbers}.
\end{theorem}
The detailed proof is provided in Appendix~\ref{sec: proof of sample complexity of unbisaed solution}. The effective rank parameter appearing in \eqref{upper bound of solution for prior knowledge summary} is specified according to the underlying state representation as
\begin{eqnarray}
\label{rank choice for different states}
\rank(\vrho^\star) =  \begin{cases}
    d^n, &\quad   \setX_{\textup{PHY}}, \setX_{\textup{S}}, \setX_{\textup{MPO}}, \setX_{\textup{PEPO}} \\
    1, &\quad  \setF_\textup{MPS}, \setF_{\textup{PEPS}}  \\
    r^{\textup{LR}}, &\quad  \setX_{\textup{LR}}, \setF_{\textup{LR}}, \setF_{\textup{LR-S}},\setF_{\textup{LR-MPO}}, \setX_{\textup{PEPO}}  \\
    \end{cases}.
\end{eqnarray}
Theorem~\ref{label:sample complexity of unbiased solution} establishes a unified non-asymptotic recovery guarantee for a broad family of structured quantum-state models under general POVM measurements. The reconstruction error is determined jointly by four complementary factors: the effective rank of the target state, the complexity of the admissible state class, the informativeness of the measurement ensemble, and the state preparation and measurement noise levels.

\begin{itemize}

\item The first factor is the effective rank of the ground-truth state. For physical states, sparse states, MPOs, and PEPOs, the theorem uses the worst-case bound $\rank(\vrho^\star)=d^n$, which corresponds to the largest possible rank within the admissible state class. Nevertheless, practical quantum systems are often approximately low rank, especially when the target state is close to a pure state or arises as the ground state of a local Hamiltonian. Consequently, the actual reconstruction error can be substantially smaller than the worst-case guarantee. For explicitly low-rank matrix models, the error scales linearly with the intrinsic rank $r^{\textup{LR}}$, while for pure-state tensor-network representations, including MPSs and PEPSs, the effective rank reduces to one, yielding the most favorable dependence on the state rank.

\item The second factor is the logarithm of the covering number, $\log N_{\textup{type}}$, which characterizes the intrinsic complexity of the underlying structured state class. The sample complexity therefore depends on the geometric complexity of the admissible model through its covering number. As summarized in Table~\ref{tab:covering-numbers}, the covering numbers of different structured quantum-state classes scale according to their intrinsic parameterizations and corresponding degrees of freedom. Consequently, the required number of state copies is determined by the intrinsic complexity of the underlying state representation rather than by the ambient Hilbert-space dimension.

\item The third factor is the normalized measurement complexity parameter $\frac{(1-a)\beta(Q,K)+a\alpha_2(Q,K)/K}
{(\alpha_1(Q,K))^2}$, which characterizes the statistical efficiency of the measurement ensemble. As established in \eqref{parameters of different POVMs}, this parameter remains of constant order for spherical $3$-design and $\delta$-approximate spherical $3$-design POVMs, whereas it decreases as $\Theta(1/Q)$ for unitary $3$-design and Haar-random projective measurements. Therefore, for randomized measurement ensembles, increasing the number of measurement settings $Q$ reduces the effective measurement complexity. As a consequence, for a fixed target reconstruction accuracy, the required number of state copies $M$ for each measurement setting can be reduced proportionally as $Q$ increases.

\item Finally, the multiplicative factor $\frac{1}{(1-a)^2(1-\lambda)^2} $ explicitly characterizes the impact of state preparation and measurement noise on the reconstruction accuracy. As the noise levels $a$ and $\lambda$ increase, the effective information contained in the measurement model decreases, resulting in a larger sample complexity for achieving the same reconstruction accuracy. In the ideal noiseless setting, namely $(a,\lambda)=(0,0)$, this factor reduces to one, and the bound reduces to the corresponding sample complexity for QST under the ideal measurement model.

\end{itemize}

Overall, Theorem~\ref{label:sample complexity of unbiased solution} provides a unified theoretical framework showing that the sample complexity of structured QST is governed jointly by the intrinsic complexity of the admissible state class, the effective rank of the target state, the quality of the measurement ensemble, and the state preparation and measurement noise levels. The theorem applies uniformly to a broad family of low-dimensional quantum-state models and general POVM measurements, thereby establishing a comprehensive characterization of the statistical complexity of structured QST.

\paragraph*{Noise-Unaware Quantum State Tomography}

In many practical quantum tomography experiments, however, the underlying noise characteristics are unavailable or only partially characterized. Consequently, most existing approaches ignore state preparation and measurement noise and reconstruct the quantum state under the ideal Born measurement model. This leads to the least-squares objective
\begin{eqnarray}
\label{loss function of biased solution}
g(\vrho) = \frac{1}{2QK}\sum_{q=1}^{Q}\sum_{k=1}^{K} \big( \<  \mA_{q,k} ,    \vrho   \> -  \wh p_{q,k}   \big)^2.
\end{eqnarray}
Although computationally simpler, the optimization model no longer matches the actual observation process whenever channel or measurement noise is present. Therefore, the resulting estimator is generally obtained under a statistically misspecified model, which introduces a systematic bias in addition to the finite-shot statistical error.  The corresponding estimator is defined as
\begin{eqnarray}
\label{minimizer rho misspecified}
\wh\vrho = \argmin_{\vrho\in\setX_{\textup{type}}} g(\vrho).
\end{eqnarray}

In view of the above discussion, many structured quantum states admit a factorized representation of the form $\vrho = \mF \mF^\dagger$ with $\mF \in {\setF}_{\textup{type}}$. Substituting this parameterization into \eqref{loss function of biased solution} yields the equivalent optimization problem
\begin{eqnarray}
\label{loss function factorized misspecified}
g(\mF) = \frac{1}{2QK}\sum_{q=1}^{Q}\sum_{k=1}^{K} \big( \langle \mA_{q,k}, \mF\mF^\dagger \rangle - \wh p_{q,k} \big)^2.
\end{eqnarray}
The corresponding estimator in the factorized parameter space is given by
\begin{eqnarray}
\label{minimizer F misspecified}
\wh \mF = \argmin_{\mF\in {\setF}_{\textup{type}}} g(\mF), \qquad \wh\vrho = \wh\mF\wh\mF^\dagger.
\end{eqnarray}
Besides the complexity parameters introduced in \Cref{subsec:Measurement Quality Characterization}, we further define the following measurement complexity parameters:
\begin{eqnarray}
    \label{lower constant of second order information cross term}
    &\!\!\!\!\!\!\!\!&\gamma_1(Q,K)  := \inf_{\vrho_1, \vrho_2 \in \setX}\frac{\sum_{q=1}^{Q}\sum_{k=1}^{K}\<\mA_{q,k}, \vrho_1\>\<\mA_{q,k}, \vrho_2\> }{\trace(\vrho_1\vrho_2)},\\
    \label{upper constant of second order information cross term}
    &\!\!\!\!\!\!\!\!&\gamma_2(Q,K) := \sup_{\vrho_1, \vrho_2 \in \setX}\frac{\sum_{q=1}^{Q}\sum_{k=1}^{K}\<\mA_{q,k}, \vrho_1\>\<\mA_{q,k}, \vrho_2\> }{\trace(\vrho_1\vrho_2)}.
\end{eqnarray}
The parameters $\gamma_1(Q,K)$ and $\gamma_2(Q,K)$ characterize the smallest and largest measurement-induced bilinear forms over pairs of admissible states, respectively. Unlike $\alpha_1(Q,K)$ and $\alpha_2(Q,K)$, which quantify the extremal quadratic forms associated with a single perturbation, $\gamma_1(Q,K)$ and $\gamma_2(Q,K)$ characterize the corresponding bilinear forms and therefore capture the interaction between two different directions in the admissible state space. In particular, setting $ \vrho_1=\vrho_2=\vrho-\vrho^\star$, reduces \eqref{lower constant of second order information cross term} and \eqref{upper constant of second order information cross term} to \eqref{lower constant of second order information} and \eqref{upper constant of second order information}, respectively. Consequently, $\gamma_1(Q,K)$ and $\gamma_2(Q,K)$ can be viewed as natural bilinear extensions of $\alpha_1(Q,K)$ and $\alpha_2(Q,K)$, respectively, and will play a key role in quantifying the model misspecification bias in the subsequent analysis. Specifically, the model misspecification bias is governed by the normalized complexity parameter
\begin{eqnarray}
\label{normalized complexity ratio}
\frac{\lambda(1-a)\gamma_2(Q,K) + a \gamma_1(Q,K)}{\alpha_1(Q,K)}.
\end{eqnarray}
To illustrate its behavior, we evaluate this quantity for several commonly used measurement ensembles. According to \cite[Lemma~8]{qin2024sample} and \cite[Eqs.~(S35) and (S36)]{huang2020predicting}, for $\setX\subset\setX_{\textup{PHY}}$, we obtain
\begin{eqnarray}
\label{parameters of different POVMs additional parameter}
&\!\!\!\!\!\!\!\!&\frac{\lambda(1-a)\gamma_2(Q,K) + a \gamma_1(Q,K)}{\alpha_1(Q,K)}\nonumber\\
&\!\!\!\!=\!\!\!\!&
\begin{dcases}
 \Theta(\lambda(1-a) + a ),
& \text{Spherical $3$-designs/Haar/Unitary $3$-designs},\\
\Theta\bigg(\frac{\lambda(1-a)(1+\delta)  + a(1-\delta)}{1-\delta}\bigg),
& \text{$\delta$-approximate spherical $3$-designs}.
\end{dcases}
\end{eqnarray}
The estimates in \eqref{parameters of different POVMs additional parameter} show that the normalized bias parameter remains of constant order for spherical $3$-designs, unitary $3$-designs and Haar random projective measurements. For $\delta$-approximate spherical $3$-designs, the additional factor depending on $\delta$ reflects the degradation caused by the approximation error and disappears as $\delta\to0$. Consequently, the model misspecification bias remains uniformly controlled for these measurement ensembles and depends primarily on the noise parameters $a$ and $\lambda$.

The following theorem establishes the corresponding sample complexity guarantee for the noise-unaware estimator.
\begin{theorem}
\label{label:sample complexity of biased solution}

Suppose $Q$ POVMs $\{\mA_{q,1},\cdots, \mA_{q,K}\}_{q\in[Q]}$ satisfy \eqref{lower constant of second order information}, \eqref{upper constant of second order information},  \eqref{upper constant of third order information}, \eqref{lower constant of second order information cross term} and \eqref{upper constant of second order information cross term}. Using each POVM to measure the quantum state $M$ times yields the empirical probabilities $\{\wh p_{q,k}\}_{q\in[Q],k\in[K]}$. Then, with probability at least $1-e^{-\Omega(\log N_{\textup{type}})}$, the estimator $\wh\vrho$ obtained from either the constrained least-squares formulation \eqref{loss function of biased solution} or its factorized counterpart \eqref{loss function factorized misspecified} satisfies
\begin{eqnarray}
\label{biased error bound of trace norm main paper}
    \|\wh{\vrho} - \vrho^\star\|_1  &\!\!\!\!\leq\!\!\!\!& O\bigg(\sqrt{\frac{((1-a)\beta(Q,K) + a\alpha_2(Q,K)/K)\rank(\vrho^\star)(\log N_{\textup{type}})}{(\alpha_1(Q,K))^2 M}}\bigg)\nonumber\\
    &\!\!\!\!\!\!\!\!& + \frac{4\sqrt{\rank(\vrho^\star)}\big( \lambda(1-a)\gamma_2(Q,K) + a\gamma_1(Q,K) \big)}{\alpha_1(Q,K)},
\end{eqnarray}
where $N_{\textup{type}}$ denotes the covering number associated with the underlying structured quantum-state class summarized in Table~\ref{tab:covering-numbers}.
\end{theorem}
The detailed proof is provided in Appendix~\ref{sec: proof of sample complexity of bisaed solution}. Compared with the recovery guarantee in Theorem~\ref{label:sample complexity of unbiased solution}, the error bound in \eqref{biased error bound of trace norm main paper} exhibits two fundamental differences.
\begin{itemize}

\item First, unlike the bound in Theorem~\ref{label:sample complexity of unbiased solution}, the multiplicative factor $\frac{1}{(1-a)^2(1-\lambda)^2} $ does not appear in the statistical error term. This difference arises because the reconstruction is performed under the ideal Born measurement model, without incorporating the state preparation and measurement noise into the forward model. Consequently, the statistical error depends only on the sampling process and retains the same order as in the noiseless setting.

\item  Second, the model misspecification introduces the additional bias term $\frac{\sqrt{\rank(\vrho^\star)}\big(\lambda(1-a)\gamma_2(Q,K)+a\gamma_1(Q,K)\big)}{\alpha_1(Q,K)} $, which is independent of the  number of total state copies $QM$. Consequently, increasing the number of measurements suppresses only the statistical error, whereas the deterministic bias persists. As a result, the reconstruction error bound contains a non-vanishing bias term unless the state preparation and measurement noise are incorporated into the reconstruction model.

\end{itemize}
Overall, \eqref{biased error bound of trace norm main paper} demonstrates that ignoring state preparation and measurement noise fundamentally changes the asymptotic behavior of the estimator: increasing the number of total state copies suppresses only the statistical error, whereas the deterministic bias persists and ultimately limits the achievable reconstruction accuracy.

\section{Numerical Experiments}
\label{sec:simulation}

We conduct numerical QST experiments using Haar-random projective measurements to validate the theoretical predictions established in the preceding sections. Unless otherwise specified, all experiments are performed on a six-qubit system with $Q=100$ Haar-random unitary measurement settings and $M=1000$ state copies per POVM. Two representative target states are considered. The first is a thermal state with temperature $T=0.2$~\cite{hartmann2004existence}, which is well approximated by low-rank and tensor-network models. The second is a GHZ state~\cite{jameson2024optimal}, whose density matrix possesses both low-rank and sparse structures and can also be represented exactly by a low-rank MPO (equivalently, an MPS since $r^{\textup{LR}}=1$). Accordingly, for the thermal state we compare the recovery performance over $\setX_{\textup{PHY}}$, $\setF_{\textup{LR}}$, $\setF_{\textup{LR\text{-}MPO}}$, and $\setF_{\textup{LR\text{-}PEPO}}$, whereas for the GHZ state we compare $\setX_{\textup{PHY}}$, $\setF_{\textup{LR}}$, $\setF_{\textup{LR\text{-}S}}$, and $\setF_{\textup{LR\text{-}MPO}}$.

Since the primary objective of this paper is to evaluate statistical performance rather than optimization algorithms, we employ first-order optimization methods for all models. Specifically, $\setX_{\textup{PHY}}$, $\setF_{\textup{LR}}$, $\setF_{\textup{LR\text{-}S}}$, and $\setF_{\textup{LR\text{-}MPO}}$ are optimized using projected gradient
descent~\cite{chen2015fast,bahmani2016learning,chen2019non,qin2024quantum,qin2024sample,qin2026structured}, where each gradient step is followed by a projection onto the corresponding structured model. Specifically, for $\setX_{\textup{PHY}}$ and $\setF_{\textup{LR}}$, the projection reduces to normalizing $\mF$ to have unit Frobenius norm. For $\setF_{\textup{LR-S}}$, the projection is implemented by retaining the $t$ entries of $\mF$ with the largest magnitudes, followed by Frobenius normalization. For $\setF_{\textup{LR-MPO}}$, the projection is performed via tensor-train SVD (TT-SVD)~\cite{oseledets2011tensor}, followed by Frobenius normalization. Since PEPOs do not admit a canonical representation~\cite{orus2019tensor}, the projection step is not available for $\setF_{\textup{LR-PEPO}}$. Instead, we parameterize the model directly by its tensor factors with a fixed bond dimension and optimize the resulting objective using gradient descent~\cite{qin2024guaranteed,qin2025scalable}. Step sizes are tuned independently for each model. We run $400$ iterations for the thermal state and $200$ iterations for the GHZ state, with target ranks $r^{\textup{LR}}=2$ and $r^{\textup{LR}}=1$, respectively. For the LR-MPO model, the bond dimension is determined adaptively by TT-SVD with tolerance $10^{-14}$ and position $B=3$ in \eqref{the set of LR-MPO states}. For the LR-PEPO model, we choose $q=2$, $p=3$, bond dimension $4$, and $B_1=1$, $B_2=2$ in \eqref{the set of LR-PEPO states}. Under the noise-aware loss \eqref{loss function factorized}, the sparsity parameter is set to $t=8$ for $a=\lambda\in\{0,0.1,0.2\}$ and to $t=16$ otherwise. Under the noise-unaware loss \eqref{loss function factorized misspecified}, we use $t=8$ only when $a=\lambda=0$ and $t=16$ otherwise. We note that although the GHZ state has true sparsity $t=2$, a larger threshold level is required in practice to achieve reliable support identification and stable optimization, consistent with previous observations for iterative hard thresholding methods~\cite{blumensath2009iterative}.

The recovery performance is evaluated by the trace-norm reconstruction error. \Cref{Thermal aware,Thermal unaware} illustrate the trace-norm reconstruction errors of different structural models for the thermal state under noise-aware and noise-unaware QST, respectively. For all four structural models, the reconstruction error increases monotonically with both the depolarizing parameter $a$ and the measurement noise level $\lambda$. Moreover, the degradation is substantially smaller under the proposed noise-aware formulation than under the conventional noise-unaware formulation, corroborating the advantage of the unbiased estimator established in our theoretical analysis. Furthermore, exploiting increasingly compact low-dimensional structures consistently improves reconstruction accuracy, with the LR-PEPO model achieving the best overall performance.

\Cref{GHZ aware,GHZ unaware} present the corresponding results for the GHZ state. Similar qualitative behavior is observed: the reconstruction error increases with both $a$ and $\lambda$, while the noise-aware formulation exhibits a smaller degradation under both state preparation and measurement noises. The LR-S model achieves the smallest reconstruction error, indicating that explicitly incorporating sparsity is particularly effective for mitigating the adverse effects of noise in GHZ-state tomography.

\begin{figure}[!ht]
\centering
\subfigure[]{
\begin{minipage}[t]{0.45\textwidth}
\centering
\includegraphics[width=7cm]{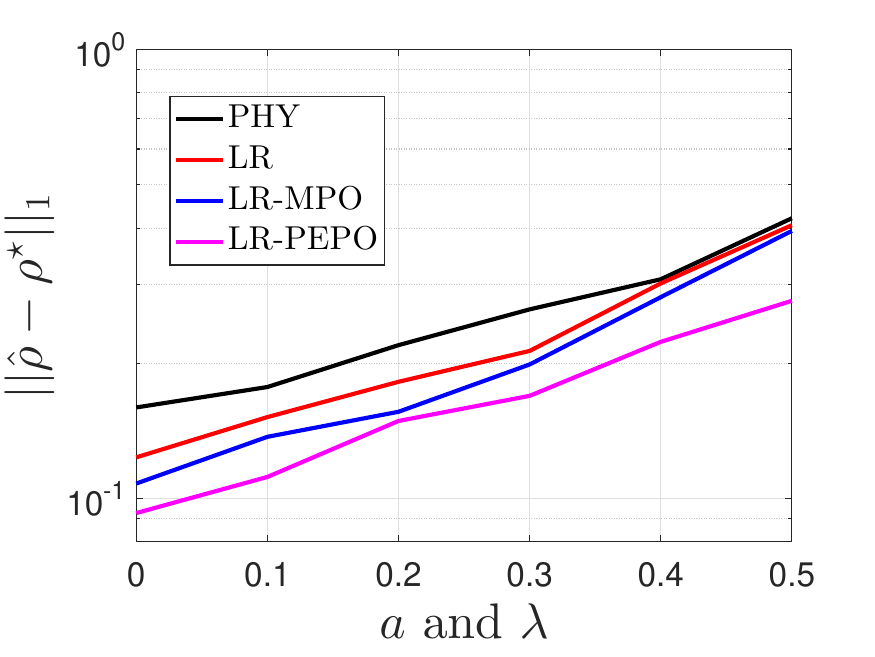}
\end{minipage}
\label{Thermal aware}
}
\subfigure[]{
\begin{minipage}[t]{0.45\textwidth}
\centering
\includegraphics[width=7cm]{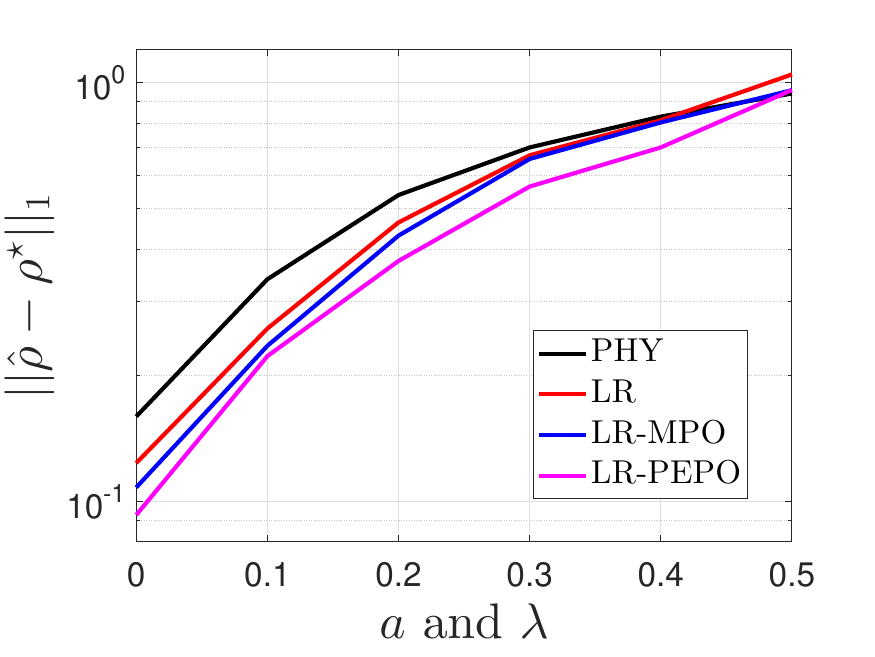}
\end{minipage}
\label{Thermal unaware}
}
\caption{Trace-norm reconstruction error for the thermal state under different values of $a$ and $\lambda$: (a) noise-aware QST, (b) noise-unaware QST.}
\end{figure}

\begin{figure}[!ht]
\centering
\subfigure[]{
\begin{minipage}[t]{0.45\textwidth}
\centering
\includegraphics[width=7cm]{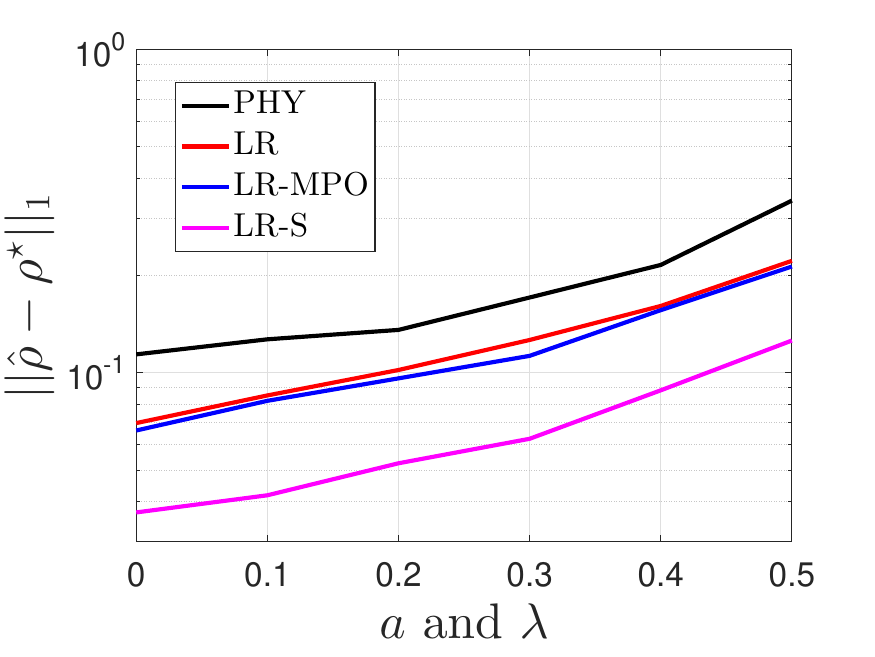}
\end{minipage}
\label{GHZ aware}
}
\subfigure[]{
\begin{minipage}[t]{0.45\textwidth}
\centering
\includegraphics[width=7cm]{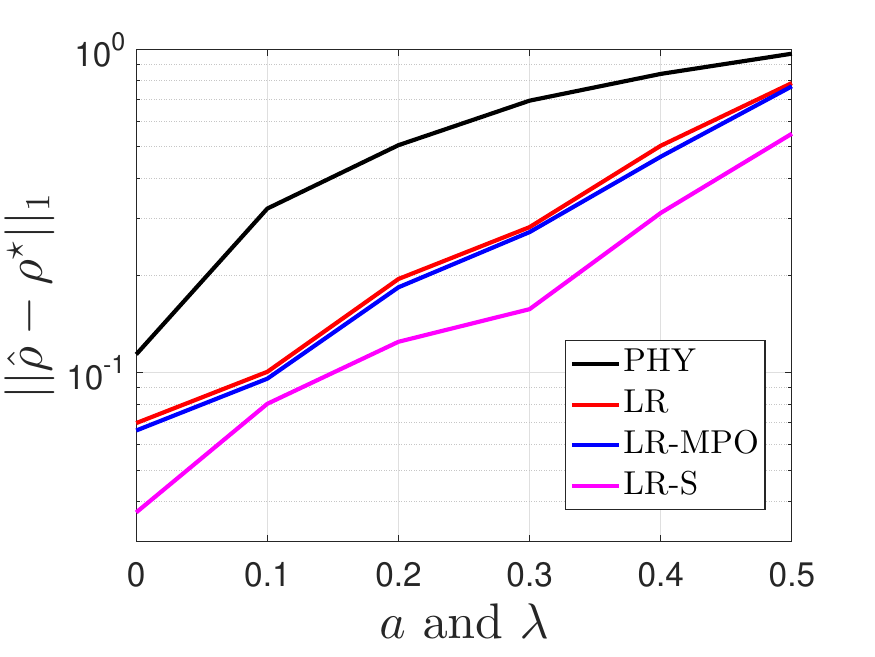}
\end{minipage}
\label{GHZ unaware}
}
\caption{Trace-norm reconstruction error for the GHZ state under different values of $a$ and $\lambda$: (a) noise-aware QST, (b) noise-unaware QST.}
\end{figure}

\section{Conclusion}
\label{sec: conclusion}

In this paper, we established a unified theoretical framework for analyzing the sample complexity of structured QST under realistic noisy observations arising from state preparation noise, measurement noise, and finite-shot statistical noise. The proposed framework applies to a broad family of structured quantum-state classes, including general mixed states, sparse states, low-rank states, MPSs, MPOs, PEPSs, PEPOs, low-rank and sparse states, LR-MPOs, and LR-PEPOs. A unified covering-number analysis was developed to characterize the intrinsic complexity of all considered structured state classes. Within this framework, we established unified non-asymptotic sample complexity guarantees for two constrained least-squares estimators. For the noise-aware estimator, we showed that the sample complexity is jointly characterized by the complexity of the underlying structured state class, the complexity of the measurement ensemble, and the state preparation and measurement noise levels. For the noise-unaware estimator, we further established a recovery guarantee consisting of both a statistical error term and an additional deterministic bias term caused by the mismatch between the reconstruction model and the noisy observation process. Unlike the statistical error, the bias term cannot be eliminated by increasing the number of state copies, thereby quantitatively characterizing the fundamental performance gap between noise-aware and noise-unaware QST.

One important direction for future research is to establish a unified theory for the necessary sample complexity of structured QST under realistic noisy observations. While this work provides unified sufficient conditions for a broad family of structured quantum-state classes, corresponding necessary conditions remain largely understood only through problem-specific techniques, such as the quantum Cram\'{e}r--Rao bound \cite{qin2026statistical}, minimax lower bounds \cite{flammia2012quantum,qin2024sample}, Holevo's theorem \cite{haah2017sample}, or Assouad's method \cite{acharya2025pauli}. Developing a unified information-theoretic framework that characterizes the fundamental sample complexity limits across different structured quantum-state classes remains an important open problem.

Another interesting direction is to extend the measurement complexity analysis to local measurement schemes \cite{lancien2013distinguishing,brandao2020fast,jameson2024optimal}. In this paper, the measurement properties are characterized primarily for global POVMs, for which explicit bounds of the associated measurement complexity parameters can be established through design-based arguments. However, local POVMs, which are experimentally more accessible and widely used in large-scale quantum systems, exhibit a fundamentally different tensor-product measurement structure. Another important open problem is to develop a unified characterization of the measurement complexity parameters for local measurement ensembles and incorporate it into the sample complexity analysis of structured QST.

Finally, the present work focuses on depolarizing noise as a representative model for state preparation and measurement imperfections. Although depolarizing noise captures a broad class of error effects and admits a tractable theoretical analysis, practical quantum devices are also affected by other imperfections, including coherent control errors~\cite{knill2008randomized,nielsen2010quantum}, relaxation and decoherence effects~\cite{nielsen2010quantum}, measurement and readout errors~\cite{maciejewski2020mitigation}, and correlated noise~\cite{clemens2004quantum}. Extending the proposed framework to these more general noise mechanisms and the corresponding quantum channels remains an important direction for future research.

\section{Acknowledgments}
\label{sec: ack}

ZQ gratefully acknowledges support from the MICDE Research Scholars Program at the University of Michigan.

\newpage

\appendices

\section{Proof of \Cref{conclusion of all covering numbers}}
\label{Proof of ocvering number all}

\begin{proof}
We establish the covering number bounds for the structured quantum-state classes introduced in \Cref{sec:structured quantum states}.
\begin{itemize}
\item \textbf{Physical states} We define the auxiliary set $\ol\setX_{\textup{PHY}} = \{ \vrho\in\C^{d^n\times d^n}:  \ \trace(\vrho) = 0, \|\vrho\|_F =  1     \}$. Fixing a collection $\{\vrho^{(1)}, \dots, \vrho^{(N_{\textup{PHY}})}\}\subset \wt\setX_{\textup{PHY}}\subset \ol\setX_{\textup{PHY}}$, we obtain an $\epsilon$-net for $\ol\setX_{\textup{PHY}}$, i.e., $\sup_{\vrho\in\ol\setX_\textup{PHY}}\min_{p\in[N_{\textup{PHY}}]}\|\vrho - \vrho^{(p)}\|_F\leq \epsilon$, whose covering number satisfies \cite{candes2011tight}
    \begin{eqnarray}
    \label{Covering number of physical sets}
    N_{\textup{PHY}}\leq \left(\frac{9}{\epsilon}\right)^{d^{2n}}.
    \end{eqnarray}
    Note that the constraint $\trace(\vrho)=0$ is not explicitly enforced in the covering-number estimate. Since it only imposes a single linear constraint on the ambient space, incorporating it would not change the asymptotic order of the covering number. The same convention is adopted throughout all subsequent covering-number estimates.

\item \textbf{Sparse states} Similarly, we define the auxiliary set $\ol\setX_{\textup{S},s} = \{ \vrho\in\C^{d^n\times d^n}:  \ \trace(\vrho) = 0, \|\vrho\|_F =  1, \|\vrho\|_{0} = s  \}$. Fixing a collection $\{\vrho^{(1)}, \dots, \vrho^{(N_{\textup{S},s})}\}\subset \wt\setX_{\textup{S},s}\subset \ol\setX_{\textup{S},s}$, we obtain an $\epsilon$-net for $\ol\setX_{\textup{S},s}$, i.e., $\sup_{\vrho\in\ol\setX_{\textup{S},s}}\min_{p\in[N_{\textup{S},s}]}\|\vrho - \vrho^{(p)}\|_F\leq \epsilon$, whose covering number satisfies \cite{vershynin2009role}
    \begin{eqnarray}
    \label{Covering number of sparse sets}
    N_{\textup{S},s}\leq \left(\frac{2Cd^{2n}}{s\epsilon}\right)^{2s},
    \end{eqnarray}
    where $C$ is a positive constant.

\item \textbf{Low-rank states} Analogously, we define the auxiliary set $\ol\setX_{\textup{LR},r^{\textup{LR}} } = \{ \vrho\in\C^{d^n\times d^n}:  \ \trace(\vrho) = 0, \|\vrho\|_F =  1,   \rank(\vrho) = r^{\textup{LR}}  \}$. For any fixed value of $\vrho^{(p)}\in \wt\setX_{\textup{LR},r^{\textup{LR}}}\subset \ol\setX_{\textup{LR},r^{\textup{LR}}}$, we can construct a $\epsilon$-net $\{\vrho^{(1)}, \dots, \vrho^{(N_{\textup{LR},r^{\textup{LR}}})}  \}$  such that $\sup_{\vrho\in\ol\setX_{\textup{LR}},r^{\textup{LR}}}\min_{p\in[N_{\textup{LR},r^{\textup{LR}}}]}\|\vrho  - \vrho^{(p)}\|_F\leq \epsilon$ with covering number \cite{candes2011tight}
    \begin{eqnarray}
    \label{Covering number of low-rank sets}
    N_{\textup{LR},r^{\textup{LR}}}\leq \left(\frac{9}{\epsilon}\right)^{(2^{n+2} + 2 )r^{\textup{LR}}}.
    \end{eqnarray}
    Note that the covering number in~\eqref{Covering number of low-rank sets} is established for the set of low-rank density matrices. Nevertheless, the same bound also applies to the set of density matrices induced by the corresponding factor set $ \ol\setF_{\textup{LR}, r^{\textup{LR}}} = \{ \mF\in\C^{d^n\times r^{\textup{LR}}}:  \ \|\mF\|_F = 1   \}$ since   factors $\mF_1,\mF_2\in \ol\setF_{\textup{LR},r^{\textup{LR}}}$ induces a low-rank matrix $ \frac{\mF_1\mF_1^\dagger- \mF_2\mF_2^\dagger}{\|\mF_1\mF_1^\dagger- \mF_2\mF_2^\dagger\|_F}\in\ol\setX_{\textup{LR}, 2r^{\textup{LR}}}$.

\item \textbf{Low-rank and sparse states} We first define the auxiliary factor set $\ol\setF_{\textup{LR-S},r^{\textup{LR}},t}= \{ \mF\in\C^{d^n\times r^{\textup{LR}}}:  \ \|\mF\|_{0} = t, \|\mF\|_F = 1   \}$ and the corresponding induced density matrix set $\ol\setX_{\textup{LR-S},r^{\textup{LR}},t^2} = \{\vrho\in\C^{d^n\times d^n}:  \ \trace(\vrho) = 0, \|\vrho\|_F =  1,   \rank(\vrho) = r^{\textup{LR}},  \|\vrho\|_0 = t^2 \}$. Since, for any $\mF_1, \mF_2 \in \ol\setF_{\textup{LR-S},r^{\textup{LR}},t}$,  $\frac{\mF_1\mF_1^\dagger- \mF_2\mF_2^\dagger}{\|\mF_1\mF_1^\dagger- \mF_2\mF_2^\dagger\|_F}\in \ol\setX_{\textup{LR-S},2r^{\textup{LR}},2t^2}$ has at most $2t^2$ nonzero entries, its covering number satisfies
    \begin{eqnarray}
    \label{Covering number of sparse and low-rank sets}
    N_{\textup{LR-S},2r^{\textup{LR}},2t^2}\leq \left(\frac{Cd^{2n}}{t^2\epsilon}\right)^{4t^2},
    \end{eqnarray}
    where $C$ is a positive constant. Although $\ol\setX_{\textup{LR-S},r^{\textup{LR}},t^2}$ also possesses a low-rank structure, we characterize its covering number using the sparsity, since the corresponding degrees of freedom $d^n r^{\textup{LR}}$ are typically much larger than $t^2$.

\item \textbf{Matrix product states (MPSs)}  We define the auxiliary set $\ol\setF_{\textup{MPS},r^{\textup{MPS}}}  = \{ \vf\in\C^{d^n\times 1}:\  \|\vf\|_2 = 1, \vf(i_1 \cdots i_\nqbit) = \mX_1^{i_1}  \cdots \mX_\nqbit^{i_\nqbit}, \mX_\ell^{i_\ell}\in\C^{r_{\ell-1}^{\textup{MPS}}\times r_\ell^{\textup{MPS}}}, \ell\in[n], r_0^{\textup{MPS}}=r_n^{\textup{MPS}}=1  \}$. Any MPS with $\|\vf\|_2 = 1$ admits an equivalent representation satisfying \cite{oseledets2011tensor,holtz2012manifolds}
    \begin{eqnarray}
    \label{Equivalent format of MPS}
    L(\mX_\ell)^\dagger L(\mX_\ell) = \mId_{r_\ell^{\textup{MPS}}}, \ell\in[n-1], \ \textup{and} \ \|L(\mX_n)\|_F = 1,
    \end{eqnarray}
    where $L(\mX_\ell)=\begin{bmatrix}\mX_{\ell}^{1} \\ \vdots\\  \mX_{\ell}^{d} \end{bmatrix}\in\C^{r_{\ell-1}^{\textup{MPS}}d \times r_\ell^{\textup{MPS}} }$. Following \cite{zhang2018tensor}, for each $\ell\in[n-1]$, one can construct an $\epsilon$-net $\{L(\mX_\ell^{(1)}),\ldots,L(\mX_\ell^{(N_\ell^{\textup{MPS},r^{\textup{MPS}}})})\}$ such that $\sup_{L(\mX_\ell): \|L(\mX_\ell)\|\leq 1}~\min_{p_\ell\in [N_\ell^{\textup{MPS},r^{\textup{MPS}}}]} \|L(\mX_\ell)-L(\mX_\ell^{(p_\ell)})\|\leq \epsilon$, with covering number
   \begin{eqnarray}
    \label{Covering number of MPS factor 1}
    N_\ell^{\textup{MPS},r^{\textup{MPS}}}\leq \left(\frac{4+\epsilon}{\epsilon}\right)^{dr_{\ell-1}^{\textup{MPS}}r_\ell^{\textup{MPS}}}.
   \end{eqnarray}
   Also, we can construct an $\epsilon$-net $\{ L(\mX_n^{(1)}), \dots, L(\mX_n^{(N_n^{\textup{MPS},r^{\textup{MPS}}})}) \}$  such that $\sup_{L(\mX_n): \|L(\mX_n)\|_F\leq 1}\min_{p_n\in[ N_n^{\textup{MPS},r^{\textup{MPS}}}]} \|L(\mX_n)-L(\mX_n^{(p_n)})\|_F\leq \epsilon$, with covering number
   \begin{eqnarray}
    \label{Covering number of MPS factor n}
    N_n^{\textup{MPS},r^{\textup{MPS}}}\leq \left(\frac{2+\epsilon}{\epsilon}\right)^{dr_{n-1}^{\textup{MPS}} }.
   \end{eqnarray}
   Therefore, by the monotonicity of covering numbers under set inclusion together with the product structure of the parameter space, the overall covering number satisfies
   \begin{eqnarray}
    \label{Covering number of MPS total}
    N_{\textup{MPS},r^{\textup{MPS}}} \leq \Pi_{\ell=1}^{n}N_\ell^{\textup{MPS},r^{\textup{MPS}}} \leq \left(\frac{4+\epsilon}{\epsilon}\right)^{\sum_{\ell=1}^{n}dr_{\ell-1}^{\textup{MPS}}r_\ell^{\textup{MPS}}}.
   \end{eqnarray}

\item \textbf{Matrix product density operators (MPOs)} We define the auxiliary set $\ol\setX_{\textup{MPO},r^{\textup{MPO}}} =  \{ \vrho\in\C^{d^n\times d^n}:  \ \trace(\vrho) =0, \|\vrho\|_F= 1,  \vrho(i_1 \cdots i_\nqbit, j_1\cdots j_\nqbit) = \mX_1^{i_1,j_1}  \cdots \mX_\nqbit^{i_\nqbit,j_\nqbit}, \mX_\ell^{i_\ell,j_\ell}\in\C^{r_{\ell-1}^{\textup{MPO}}\times r_\ell^{\textup{MPO}}}, \ell\in[n], r_0^{\textup{MPO}}=r_n^{\textup{MPO}}=1   \}$. As in the MPS case, any MPO with $\|\vrho\|_F= 1$ admits an equivalent representation satisfying
    \begin{eqnarray}
    \label{Equivalent format of MPO}
    L(\mX_\ell)^\dagger L(\mX_\ell) = \mId_{r_\ell^{\textup{MPO}}}, \ell\in[n-1], \ \textup{and} \ \|L(\mX_n)\|_F = 1,
    \end{eqnarray}
    where $L(\mX_\ell)=\begin{bmatrix}\mX_{\ell}^{1,1} \\ \vdots\\  \mX_{\ell}^{d,d} \end{bmatrix}\in\C^{r_{\ell-1}^{\textup{MPO}}d^2 \times r_\ell^{\textup{MPO}} }$. By the same argument as for MPS, we obtain the following covering number:
   \begin{eqnarray}
    \label{Covering number of MPO total}
    N_{\textup{MPO},r^{\textup{MPO}}} \leq  \left(\frac{4+\epsilon}{\epsilon}\right)^{\sum_{\ell=1}^{n}d^2r_{\ell-1}^{\textup{MPO}}r_\ell^{\textup{MPO}}}.
   \end{eqnarray}

\item \textbf{Low-rank MPOs (LR-MPOs)} We define the corresponding auxiliary set $\ol\setF_{\textup{LR-MPO}, r^{\textup{LR}}, r^{\textup{LR-MPO}}} = \{ \mF\in\C^{d^n\times r^{\textup{LR}}}:\ \|\mF\|_F = 1, \mF(i_1 \cdots i_\nqbit,j) = \mX_1^{i_1} \cdots \mX_{B}^{i_{B},j} \cdots \mX_\nqbit^{i_\nqbit}, j\in[r^{\textup{LR}}],   \mX_\ell^{i_\ell}\in\C^{r_{\ell-1}^{\textup{LR-MPO}}\times r_\ell^{\textup{LR-MPO}}}, \ell\in[n]\setminus\{B\} \  \textup{and} \  \mX_B^{i_B,j}\in\C^{r_{B-1}^{\textup{LR-MPO}}\times r_B^{\textup{LR-MPO}}}, \ r_0^{\textup{LR-MPO}}=r_n^{\textup{LR-MPO}}=1 \}$. Observe that every matrix $\frac{\mF_1\mF_1^\dagger- \mF_2\mF_2^\dagger}{\|\mF_1\mF_1^\dagger- \mF_2\mF_2^\dagger\|_F}$, $\mF_1, \mF_2\in\ol\setF_{\textup{LR-MPO}, r^{\textup{LR}}, r^{\textup{LR-MPO}}}$ belongs to the set $\ol\setX_{\textup{LR-MPO}, 2r^{\textup{LR}}, 2r^{\textup{LR-MPO}}} = \{ \vrho\in\C^{d^n\times d^n}:  \ \trace(\vrho) = 0, \|\vrho\|_F = 1, \rank(\vrho) = 2r^{\textup{LR}}, \vrho(i_1 \cdots i_\nqbit, j_1\cdots j_\nqbit) = \mX_1^{i_1,j_1}  \cdots \mX_\nqbit^{i_\nqbit,j_\nqbit}, \mX_\ell^{i_\ell,j_\ell}\in\C^{2(r_{\ell-1}^{\textup{LR-MPO}})^2\times 2(r_\ell^{\textup{LR-MPO}})^2}, \ell\in[n]    \}$. Therefore, applying the covering number bound for MPO yields
   \begin{eqnarray}
    \label{Covering number of LR-MPO total}
    N_{\textup{LR-MPO}, 2r^{\textup{LR}}, 2r^{\textup{LR-MPO}}} \leq  \left(\frac{4+\epsilon}{\epsilon}\right)^{\sum_{\ell=1}^{n}4d^2(r_{\ell-1}^{\textup{LR-MPO}})^2 (r_\ell^{\textup{LR-MPO}})^2}.
   \end{eqnarray}

\item \textbf{Projected entangled pair states (PEPSs)} For notational simplicity, let $\vf = [\mX_{\ul{1}\,\ul{1}}, \dots, \mX_{\ul{q}\,\ul{p}} ]$ denote the collection of all PEPSs. Unlike MPS/MPO representations, PEPSs generally do not admit a canonical form due to the presence of loops in the underlying tensor network. Consequently, bounds on the local tensor norms alone do not imply a corresponding bound on the norm of the represented state, and vice versa. To ensure both local parameter regularity and global boundedness of the represented state, we introduce the following auxiliary parameter set $\ol\setF_{\textup{PEPS},r^{\textup{PEPS}}} = \{ \vf\in\C^{d^n\times 1}: \vf = [\mX_{\ul{1}\,\ul{1}}, \dots, \mX_{\ul{q}\,\ul{p}} ], \|\mX_{\ul{a}\,\ul{b}} \|_F\leq L_{\ul{a}\,\ul{b}}, \mX_{\ul{a}\,\ul{b}}\in\C^{d\times r_{\ul{a}\,\ul{b-1},\,\ul{a}\,\ul{b}}^{\textup{PEPS}}\times r_{\ul{a-1}\,\ul{b},\,\ul{a}\,\ul{b}}^{\textup{PEPS}} \times r_{\ul{a}\,\ul{b},\,\ul{a}\,\ul{b+1}}^{\textup{PEPS}}\times r_{\ul{a}\,\ul{b},\,\ul{a+1}\,\ul{b}}^{\textup{PEPS}}}, \ul{a}\in[q],  \ul{b}\in[p]  \} \cap \{\vf\in\C^{d^n\times 1}: \|\vf\|_2\leq L  \}$. For each local tensor $\mX_{\ul{a}\,\ul{b}}$, there exists an $\epsilon_{\ul{a}\,\ul{b}}$-net $\{\mX_{\ul{a}\,\ul{b}}^{(1)}, \dots, \mX_{\ul{a}\,\ul{b}}^{(N_{\ul{a}\,\ul{b}}^{\textup{PEPS},r^{\textup{PEPS}}})}  \}$ such that $\sup_{\mX_{\ul{a}\,\ul{b}}: \|\mX_{\ul{a}\,\ul{b}} \|_F\leq L_{\ul{a}\,\ul{b}}}\min_{\ell_{\ul{a}\,\ul{b}}\in[ N_{\ul{a}\,\ul{b}}^{\textup{PEPS},r^{\textup{PEPS}}}]} \|\mX_{\ul{a}\,\ul{b}}-\mX_{\ul{a}\,\ul{b}}^{(\ell_{\ul{a}\,\ul{b}})}\|_F\leq L_{\ul{a}\,\ul{b}} \epsilon_{\ul{a}\,\ul{b}}$  whose covering number satisfies
    \begin{eqnarray}
    \label{Covering number of PEPS factor  any}
    N_{\ul{a}\,\ul{b}}^{\textup{PEPS},r^{\textup{PEPS}}} \leq \left(\frac{2+\epsilon_{\ul{a}\,\ul{b}}}{\epsilon_{\ul{a}\,\ul{b}}}\right)^{dr_{\ul{a}\,\ul{b-1},\,\ul{a}\,\ul{b}}^{\textup{PEPS}} r_{\ul{a-1}\,\ul{b},\,\ul{a}\,\ul{b}}^{\textup{PEPS}} r_{\ul{a}\,\ul{b},\,\ul{a}\,\ul{b+1}}^{\textup{PEPS}} r_{\ul{a}\,\ul{b},\,\ul{a+1}\,\ul{b}}^{\textup{PEPS}}}.
   \end{eqnarray}
   Likewise, there exists an $\epsilon$-net $\{ \vf^{(1)}, \dots, \vf^{(N_{\vf}^{\textup{PEPS}})}  \}$ such that $\sup_{\vf: \|\vf \|_2\leq 1}\min_{\ell\in[ N_{\vf}^{\textup{PEPS}}]} \|\vf-\vf^{(\ell)}\|_2\leq \epsilon $, with covering number
    \begin{eqnarray}
    \label{Covering number of PEPS entire vector}
    N_{\vf}^{\textup{PEPS}} \leq \left(\frac{2L+\epsilon}{\epsilon}\right)^{d^n}.
    \end{eqnarray}
    Both covering constructions provide valid coverings of the auxiliary parameter set $\ol\setF_{\textup{PEPS},r^{\textup{PEPS}}}$. The first is induced by the product structure of the local tensors, while the second is obtained by viewing $\vf$ as a vector in the ambient space. Therefore, the covering number of $\ol\setF_{\textup{PEPS},r^{\textup{PEPS}}}$ is bounded by the smaller of the two upper bounds, namely,
   \begin{eqnarray}
    \label{Covering number of PEPS total}
    N_{\textup{PEPS},r^{\textup{PEPS}}} &\!\!\!\!\leq\!\!\!\!& \min\{ \Pi_{\ul{a}=1}^{q}\Pi_{\ul{b}=1}^{p} N_{\ul{a}\,\ul{b}}^{\textup{PEPS},r^{\textup{PEPS}}},  N_{\vf}^{\textup{PEPS}} \}\nonumber\\
     &\!\!\!\!=\!\!\!\!& \Pi_{\ul{a}=1}^{q}\Pi_{\ul{b}=1}^{p}\left(\frac{2+\epsilon_{\ul{a}\,\ul{b}}}{\epsilon_{\ul{a}\,\ul{b}}}\right)^{dr_{\ul{a}\,\ul{b-1},\,\ul{a}\,\ul{b}}^{\textup{PEPS}} r_{\ul{a-1}\,\ul{b},\,\ul{a}\,\ul{b}}^{\textup{PEPS}} r_{\ul{a}\,\ul{b},\,\ul{a}\,\ul{b+1}}^{\textup{PEPS}} r_{\ul{a}\,\ul{b},\,\ul{a+1}\,\ul{b}}^{\textup{PEPS}}}.
   \end{eqnarray}
   In typical PEPS parameterizations, the product-structure bound is tighter than the ambient-space bound, and we therefore adopt the former in the following analysis.

\item \textbf{Projected entangled pair operators (PEPOs)} Similar to the PEPS case, for notational simplicity, we represent a PEPO by $\vrho = [\mX_{\ul{1}\,\ul{1}}, \dots, \mX_{\ul{q}\,\ul{p}} ]$ where the entries denote the local PEPO tensors. We then define the auxiliary parameter set  $\ol\setX_{\textup{PEPO},r^{\textup{PEPO}}} = \{ \vrho\in\C^{d^n\times d^n}: \vrho = [\mX_{\ul{1}\,\ul{1}}, \dots, \mX_{\ul{q}\,\ul{p}} ], \|\mX_{\ul{a}\,\ul{b}} \|_F\leq L_{\ul{a}\,\ul{b}}, \mX_{\ul{a}\,\ul{b}}\in\C^{d\times d\times r_{\ul{a}\,\ul{b-1},\,\ul{a}\,\ul{b}}^{\textup{PEPO}}\times r_{\ul{a-1}\,\ul{b},\,\ul{a}\,\ul{b}}^{\textup{PEPO}} \times r_{\ul{a}\,\ul{b},\,\ul{a}\,\ul{b+1}}^{\textup{PEPO}}\times r_{\ul{a}\,\ul{b},\,\ul{a+1}\,\ul{b}}^{\textup{PEPO}}}, \ul{a}\in[q],  \ul{b}\in[p]   \} \cap \{\vrho\in\C^{d^n\times d^n}: \|\vrho\|_F\leq L  \}$. By an argument entirely analogous to that for PEPS, the covering number of $\ol\setX_{\textup{PEPO}}$ satisfies
   \begin{eqnarray}
    \label{Covering number of PEPO total}
    N_{\textup{PEPO},r^{\textup{PEPO}}} \leq \Pi_{\ul{a}=1}^{q}\Pi_{\ul{b}=1}^{p}\left(\frac{2+\epsilon_{\ul{a}\,\ul{b}}}{\epsilon_{\ul{a}\,\ul{b}}}\right)^{d^2r_{\ul{a}\,\ul{b-1},\,\ul{a}\,\ul{b}}^{\textup{PEPO}} r_{\ul{a-1}\,\ul{b},\,\ul{a}\,\ul{b}}^{\textup{PEPO}} r_{\ul{a}\,\ul{b},\,\ul{a}\,\ul{b+1}}^{\textup{PEPO}} r_{\ul{a}\,\ul{b},\,\ul{a+1}\,\ul{b}}^{\textup{PEPO}}}.
   \end{eqnarray}

\item \textbf{Low-rank PEPOs (LR-PEPOs)} Similar to the LR-MPO construction, for every $\mF_1, \mF_2\in\setF_{\textup{LR-PEPO}}$, we define $ \vrho = \mF_1\mF_1^\dagger - \mF_2\mF_2^\dagger$. Then $\vrho$ admits a PEPO representation whose bond dimensions are $[2(r_{\ul{1}\,\ul{1},\,\ul{1}\,\ul{2}}^{\textup{LR-PEPO}})^2,\dots, 2(r_{\ul{q}\,\ul{p-1},\,\ul{q}\,\ul{p}}^{\textup{LR-PEPO}})^2 ]$. Applying the covering number bound for PEPOs therefore yields
   \begin{eqnarray}
    \label{Covering number of LR-PEPO total}
    N_{\textup{LR-PEPO}, 2r^{\textup{LR}}, 2r^{\textup{LR-PEPO}}} \leq \Pi_{\ul{a}=1}^{q}\Pi_{\ul{b}=1}^{p}\left(\frac{2+\epsilon_{\ul{a}\,\ul{b}}}{\epsilon_{\ul{a}\,\ul{b}}}\right)^{16d^2 (r_{\ul{a}\,\ul{b-1},\,\ul{a}\,\ul{b}}^{\textup{LR-PEPO}})^2 (r_{\ul{a-1}\,\ul{b},\,\ul{a}\,\ul{b}}^{\textup{LR-PEPO}})^2 (r_{\ul{a}\,\ul{b},\,\ul{a}\,\ul{b+1}}^{\textup{LR-PEPO}})^2 (r_{\ul{a}\,\ul{b},\,\ul{a+1}\,\ul{b}}^{\textup{LR-PEPO}})^2}.
   \end{eqnarray}
   Here, $\epsilon_{\ul{a}\,\ul{b}}$ is defined in the same way as in \eqref{Covering number of PEPO total}.

\end{itemize}
For the subsequent analysis, we choose $\epsilon=\frac{1}{2}$ for physical states, $\epsilon=\frac{1}{4}$ for sparse, low-rank, and low-rank and sparse states, $\epsilon=\frac{1}{4n}$ for MPSs, and $\epsilon=\frac{1}{2n}$ for MPOs and LR-MPOs. Likewise, we choose $\epsilon_{\ul{a},\ul{b}}=\frac{1}{4n}$ for PEPSs and $\epsilon_{\ul{a},\ul{b}}=\frac{1}{2n}$ for PEPOs and LR-PEPOs. This completes the proof.

\end{proof}

\section{Proof of \Cref{label:sample complexity of unbiased solution}}
\label{sec: proof of sample complexity of unbisaed solution}

Since every factor $\mF\in\setF$ induces a quantum state through $\vrho=\mF\mF^\dagger $, we equivalently consider the induced state class $ \setX = \{\vrho\in\mathbb{C}^{d^n\times d^n}: \vrho=\mF\mF^\dagger,\; \mF\in\setF\}$. In the following, we work directly with the induced density operator $\vrho\in\setX$. Since $\wh\vrho$ is a global minimizer of \eqref{minimizer loss function of unbiased solution}, it follows that
\begin{eqnarray}
\label{necessary condition of global solution}
f(\wh\vrho) \leq f(\vrho^\star).
\end{eqnarray}
By expanding \eqref{necessary condition of global solution} and using
$\wh p_{q,k}=p_{q,k}+\eta_{q,k}$, we obtain
\begin{eqnarray}
\label{necessary condition of global solution1}
&\!\!\!\!\!\!\!\!&\frac{1}{2QK}\sum_{q=1}^{Q}\sum_{k=1}^{K}\bigg\<(1-a) \mA_{q,k} + \frac{a}{K}\mId,   (1-\lambda)(\wh\vrho - \vrho^\star)   \bigg\>^2 \nonumber\\
&\!\!\!\!\leq\!\!\!\!& \frac{1}{QK} \sum_{q=1}^{Q}\sum_{k=1}^{K}\eta_{q,k}\bigg\<(1-a) \mA_{q,k} + \frac{a}{K}\mId,   (1-\lambda)(\wh\vrho - \vrho^\star)   \bigg\>,
\end{eqnarray}
where $\eta_{q,k} = \wh p_{q,k} - p_{q,k} = \frac{f_{q,k}}{M} - p_{q,k}$ denotes the finite-shot statistical error associated with the empirical probability.

We first derive a lower bound for the left-hand side of \eqref{necessary condition of global solution1}. Expanding the square yields
\begin{eqnarray}
\label{lower bound of the left term}
&\!\!\!\!\!\!\!\!&\sum_{q=1}^{Q}\sum_{k=1}^{K}\bigg\<(1-a) \mA_{q,k} + \frac{a}{K}\mId,   (1-\lambda)(\wh\vrho - \vrho^\star)   \bigg\>^2\nonumber\\
&\!\!\!\! = \!\!\!\!& \sum_{q=1}^{Q}\sum_{k=1}^{K}(1-a)^2(1-\lambda)^2\<\mA_{q,k}, \wh\vrho - \vrho^\star \>^2 + \sum_{q=1}^{Q}\sum_{k=1}^{K}\frac{a^2(1-\lambda)^2}{K^2}\< \mId, \wh\vrho - \vrho^\star \>^2\nonumber\\
&\!\!\!\!\!\!\!\!& + \sum_{q=1}^{Q}\sum_{k=1}^{K} \frac{2a(1-a)(1-\lambda)^2}{K}\< \mA_{q,k}, \wh\vrho - \vrho^\star \>\< \mId, \wh\vrho - \vrho^\star  \>\nonumber\\
&\!\!\!\! = \!\!\!\!& \sum_{q=1}^{Q}\sum_{k=1}^{K}(1-a)^2(1-\lambda)^2\<\mA_{q,k}, \wh\vrho - \vrho^\star \>^2\nonumber\\
&\!\!\!\! \geq  \!\!\!\!& \alpha_1(Q,K) (1-a)^2(1-\lambda)^2 \|\wh\vrho - \vrho^\star\|_F^2,
\end{eqnarray}
where the second equality follows from $\< \mId, \wh\vrho - \vrho^\star  \> = \trace(\wh\vrho - \vrho^\star ) = 0$ and the last inequality follows from \eqref{lower constant of second order information}.

We next analyze the right-hand side of \eqref{necessary condition of global solution1}. By direct expansion,
\begin{eqnarray}
\label{upper bound of the right term}
\sum_{q=1}^{Q}\sum_{k=1}^{K}\eta_{q,k}\bigg\<(1-a) \mA_{q,k} + \frac{a}{K}\mId,   (1-\lambda)(\wh\vrho - \vrho^\star)   \bigg\>  =   (1-a)(1-\lambda)\sum_{q=1}^{Q}\sum_{k=1}^{K}\eta_{q,k} \< \mA_{q,k}, \wh\vrho - \vrho^\star   \>,
\end{eqnarray}
where the equality follows from $\< \mId, \wh\vrho - \vrho^\star  \> = \trace(\wh\vrho - \vrho^\star ) = 0$. The subsequent analysis is carried out separately for the following structured quantum-state classes: physical/low-rank/sparse states, MPOs, and PEPOs.
\begin{itemize}

\item \textbf{Physical states} Let  $ \{\vrho^{(1)}, \dots, \vrho^{(N_{\textup{PHY}})}\} \subset\ol\setX_{\textup{PHY}}$ be the $\epsilon$-net constructed in Appendix~\ref{Proof of ocvering number all}, and fix an arbitrary covering point $\vrho^{(p)}\in \{\vrho^{(1)}, \dots, \vrho^{(N_{\textup{PHY}})}\}$. Then,
    \begin{eqnarray}
    \label{upper bound of the right term physical}
    &\!\!\!\! \!\!\!\!&\sum_{q=1}^{Q}\sum_{k=1}^{K}\eta_{q,k} \< \mA_{q,k}, \wh\vrho - \vrho^\star   \>\nonumber\\
    &\!\!\!\! \leq \!\!\!\!& \|\wh\vrho - \vrho^\star\|_F\max_{\vrho\in\ol\setX_{\textup{PHY}}}\sum_{q=1}^{Q}\sum_{k=1}^{K}\eta_{q,k} \< \mA_{q,k}, \vrho  \>\nonumber\\
    &\!\!\!\! = \!\!\!\!& \|\wh\vrho - \vrho^\star\|_F\max_{\vrho\in\ol\setX_{\textup{PHY}}}\sum_{q=1}^{Q}\sum_{k=1}^{K}\eta_{q,k} \< \mA_{q,k}, \vrho - \vrho^{(p)}  \> + \|\wh\vrho - \vrho^\star\|_F\sum_{q=1}^{Q}\sum_{k=1}^{K}\eta_{q,k} \< \mA_{q,k}, \vrho^{(p)}  \>\nonumber\\
    &\!\!\!\! \leq \!\!\!\!& \|\wh\vrho - \vrho^\star\|_F \|\vrho - \vrho^{(p)}\|_F \max_{\vrho\in\ol\setX_{\textup{PHY}}}\sum_{q=1}^{Q}\sum_{k=1}^{K}\eta_{q,k} \< \mA_{q,k}, \frac{\vrho - \vrho^{(p)}}{\|\vrho - \vrho^{(p)}\|_F}   \> + \|\wh\vrho - \vrho^\star\|_F\sum_{q=1}^{Q}\sum_{k=1}^{K}\eta_{q,k} \< \mA_{q,k}, \vrho^{(p)}  \>\nonumber\\
    &\!\!\!\! \leq \!\!\!\!& \|\wh\vrho - \vrho^\star\|_F \epsilon \max_{\vrho\in\ol\setX_{\textup{PHY}}}\sum_{q=1}^{Q}\sum_{k=1}^{K}\eta_{q,k} \< \mA_{q,k}, \vrho   \> + \|\wh\vrho - \vrho^\star\|_F\sum_{q=1}^{Q}\sum_{k=1}^{K}\eta_{q,k} \< \mA_{q,k}, \vrho^{(p)}  \>.
    \end{eqnarray}
    Setting $\epsilon=\frac12$ yields
    \begin{eqnarray}
    \label{upper bound of the right term physical1}
    \sum_{q=1}^{Q}\sum_{k=1}^{K}\eta_{q,k} \< \mA_{q,k}, \wh\vrho - \vrho^\star   \> &\!\!\!\! \leq \!\!\!\!& \|\wh\vrho - \vrho^\star\|_F\max_{\vrho\in\ol\setX_{\textup{PHY}}}\sum_{q=1}^{Q}\sum_{k=1}^{K}\eta_{q,k} \< \mA_{q,k}, \vrho  \>\nonumber\\
    &\!\!\!\! \leq \!\!\!\!& 2\|\wh\vrho - \vrho^\star\|_F\sum_{q=1}^{Q}\sum_{k=1}^{K}\eta_{q,k} \< \mA_{q,k}, \vrho^{(p)}  \>.
    \end{eqnarray}

\item \textbf{Sparse states} Let  $ \{\vrho^{(1)}, \dots, \vrho^{(N_{\textup{S},2s})}\} \subset\ol\setX_{\textup{S},2s}$ be the $\epsilon$-net constructed in Appendix~\ref{Proof of ocvering number all}, and fix an arbitrary covering point $\vrho^{(p)}\in \{\vrho^{(1)}, \dots, \vrho^{(N_{\textup{S},2s})}\}$. Then,
    \begin{eqnarray}
    \label{upper bound of the right term sparse}
    &\!\!\!\! \!\!\!\!&\sum_{q=1}^{Q}\sum_{k=1}^{K}\eta_{q,k} \< \mA_{q,k}, \wh\vrho - \vrho^\star   \>  \leq   \|\wh\vrho - \vrho^\star\|_F\max_{\vrho\in\ol\setX_{\textup{S},2s}}\sum_{q=1}^{Q}\sum_{k=1}^{K}\eta_{q,k} \< \mA_{q,k}, \vrho  \>\nonumber\\
    &\!\!\!\! = \!\!\!\!& \|\wh\vrho - \vrho^\star\|_F\max_{\vrho\in\ol\setX_{\textup{S},2s}}\sum_{q=1}^{Q}\sum_{k=1}^{K}\eta_{q,k} \< \mA_{q,k}, \vrho - \vrho^{(p)}  \> + \|\wh\vrho - \vrho^\star\|_F\sum_{q=1}^{Q}\sum_{k=1}^{K}\eta_{q,k} \< \mA_{q,k}, \vrho^{(p)}  \>\nonumber\\
    &\!\!\!\! \leq \!\!\!\!& \|\wh\vrho - \vrho^\star\|_F\max_{\vrho_1,\vrho_2\in\ol\setX_{\textup{S},2s}}\sum_{q=1}^{Q}\sum_{k=1}^{K}\eta_{q,k} (\< \mA_{q,k}, \vrho_1 - \vrho_1^{(p)}  \> + \< \mA_{q,k}, \vrho_2 - \vrho_2^{(p)}  \> )  + \|\wh\vrho - \vrho^\star\|_F\sum_{q=1}^{Q}\sum_{k=1}^{K}\eta_{q,k} \< \mA_{q,k}, \vrho^{(p)}  \>\nonumber\\
    &\!\!\!\! \leq \!\!\!\!& \|\wh\vrho - \vrho^\star\|_F\max_{\vrho_1,\vrho_2\in\ol\setX_{\textup{S},2s}}\sum_{q=1}^{Q}\sum_{k=1}^{K}\eta_{q,k} (\< \mA_{q,k}, \vrho_1 - \vrho_1^{(p)}  \> + \< \mA_{q,k}, \vrho_2 - \vrho_2^{(p)}  \> )\nonumber\\
    &\!\!\!\!  \!\!\!\!& + \|\wh\vrho - \vrho^\star\|_F\sum_{q=1}^{Q}\sum_{k=1}^{K}\eta_{q,k} \< \mA_{q,k}, \vrho^{(p)}  \>\nonumber\\
    &\!\!\!\! \leq \!\!\!\!& \|\wh\vrho - \vrho^\star\|_F\max_{\vrho_1,\vrho_2\in\ol\setX_{\textup{S},2s}}\sum_{q=1}^{Q}\sum_{k=1}^{K}\eta_{q,k} \bigg(\|\vrho_1 - \vrho_1^{(p)}\|_F\bigg\< \mA_{q,k}, \frac{\vrho_1 - \vrho_1^{(p)}}{\|\vrho_1 - \vrho_1^{(p)}\|_F}  \bigg\>\nonumber\\
    &\!\!\!\!  \!\!\!\!& + \|\vrho_2 - \vrho_2^{(p)}\|_F\bigg\< \mA_{q,k},\frac{\vrho_2 - \vrho_2^{(p)}}{\|\vrho_2 - \vrho_2^{(p)}\|_F}  \bigg\> \bigg)
    + \|\wh\vrho - \vrho^\star\|_F\sum_{q=1}^{Q}\sum_{k=1}^{K}\eta_{q,k} \< \mA_{q,k}, \vrho^{(p)}  \>\nonumber\\
    &\!\!\!\! \leq \!\!\!\!&  \|\wh\vrho - \vrho^\star\|_F 2\epsilon\max_{\vrho\in\ol\setX_{\textup{S},2s}}\sum_{q=1}^{Q}\sum_{k=1}^{K}\eta_{q,k} \< \mA_{q,k}, \vrho   \> + \|\wh\vrho - \vrho^\star\|_F\sum_{q=1}^{Q}\sum_{k=1}^{K}\eta_{q,k} \< \mA_{q,k}, \vrho^{(p)}  \>.
    \end{eqnarray}
    Setting $\epsilon=\frac{1}{4}$ yields
    \begin{eqnarray}
    \label{upper bound of the right term sparse1}
    \sum_{q=1}^{Q}\sum_{k=1}^{K}\eta_{q,k} \< \mA_{q,k}, \wh\vrho - \vrho^\star   \> &\!\!\!\! \leq \!\!\!\!& \|\wh\vrho - \vrho^\star\|_F\max_{\vrho\in\ol\setX_{\textup{S},2s}}\sum_{q=1}^{Q}\sum_{k=1}^{K}\eta_{q,k} \< \mA_{q,k}, \vrho  \>\nonumber\\
    &\!\!\!\! \leq \!\!\!\!& 2\|\wh\vrho - \vrho^\star\|_F\sum_{q=1}^{Q}\sum_{k=1}^{K}\eta_{q,k} \< \mA_{q,k}, \vrho^{(p)}  \>.
    \end{eqnarray}

\item \textbf{Low-rank states} Let  $ \{\vrho^{(1)}, \dots, \vrho^{(N_{\textup{LR},2r^{\textup{LR}}})}\} \subset\ol\setX_{\textup{LR},2r^{\textup{LR}}}$ be the $\epsilon$-net constructed in Appendix~\ref{Proof of ocvering number all}, and fix an arbitrary covering point $\vrho^{(p)}\in \{\vrho^{(1)}, \dots, \vrho^{(N_{\textup{LR},2r^{\textup{LR}}})}\}$. Applying the same argument as in the derivation of \eqref{upper bound of the right term sparse1} with $\epsilon=\frac{1}{4}$ yields
    \begin{eqnarray}
    \label{upper bound of the right term low-rank}
    \sum_{q=1}^{Q}\sum_{k=1}^{K}\eta_{q,k} \< \mA_{q,k}, \wh\vrho - \vrho^\star   \>  \leq   2\|\wh\vrho - \vrho^\star\|_F\sum_{q=1}^{Q}\sum_{k=1}^{K}\eta_{q,k} \< \mA_{q,k}, \vrho^{(p)}  \>.
    \end{eqnarray}

\item \textbf{MPOs} Let  $ \{\vrho^{(1)}, \dots, \vrho^{(N_{\textup{MPO},2r^{\textup{MPO}}})}\} \subset\ol\setX_{\textup{MPO},2r^{\textup{MPO}}}$ be the $\epsilon$-net constructed in Appendix~\ref{Proof of ocvering number all}, and fix an arbitrary covering point $\vrho^{(p)}\in \{\vrho^{(1)}, \dots, \vrho^{(N_{\textup{MPO},2r^{\textup{MPO}}})}\}$. To simplify the notation, we write the MPO representation of $\vrho$ compactly as $\vrho = [\mX_1,\dots, \mX_n] $ and then
    \begin{eqnarray}
    \label{upper bound of the right term MPO}
    &\!\!\!\! \!\!\!\!&\sum_{q=1}^{Q}\sum_{k=1}^{K}\eta_{q,k} \< \mA_{q,k}, \wh\vrho - \vrho^\star   \>\nonumber\\
    &\!\!\!\! \leq \!\!\!\!& \|\wh\vrho - \vrho^\star\|_F\max_{\vrho\in\ol\setX_{\textup{MPO},2r^{\textup{MPO}}}}\sum_{q=1}^{Q}\sum_{k=1}^{K}\eta_{q,k} \< \mA_{q,k}, \vrho  \>\nonumber\\
    &\!\!\!\! = \!\!\!\!& \|\wh\vrho - \vrho^\star\|_F\max_{\vrho\in\ol\setX_{\textup{MPO},2r^{\textup{MPO}}}}\sum_{q=1}^{Q}\sum_{k=1}^{K}\eta_{q,k} \< \mA_{q,k}, \vrho - \vrho^{(p)}  \> + \|\wh\vrho - \vrho^\star\|_F\sum_{q=1}^{Q}\sum_{k=1}^{K}\eta_{q,k} \< \mA_{q,k}, \vrho^{(p)}  \>\nonumber\\
    &\!\!\!\! \leq \!\!\!\!& \|\wh\vrho - \vrho^\star\|_F\max_{[\mX_1,\dots, \mX_n]\in\ol\setX_{\textup{MPO},2r^{\textup{MPO}}}}\sum_{q=1}^{Q}\sum_{k=1}^{K}\eta_{q,k} \bigg\< \mA_{q,k}, \sum_{a=1}^n[\mX_1^{(p)},\dots,\! \mX_{a-1}^{(p)}, \mX_{a}^{(p)}\! -\! \mX_{a}^\star, \mX_{a+1}^\star,  \dots, \! \mX_n^\star]  \bigg\> \nonumber\\
    &\!\!\!\!  \!\!\!\!& + \|\wh\vrho - \vrho^\star\|_F\sum_{q=1}^{Q}\sum_{k=1}^{K}\eta_{q,k} \< \mA_{q,k}, \vrho^{(p)}  \>\nonumber\\
    &\!\!\!\! \leq \!\!\!\!& \|\wh\vrho - \vrho^\star\|_F \epsilon\max_{[\mX_1,\dots, \mX_n]\in\ol\setX_{\textup{MPO},2r^{\textup{MPO}}}}\sum_{q=1}^{Q}\sum_{k=1}^{K}\eta_{q,k} \bigg\< \mA_{q,k}, \sum_{a=1}^n[\mX_1^{(p)},\dots,\! \mX_{a-1}^{(p)}, \frac{\mX_{a}^{(p)}\! -\! \mX_{a}^\star}{\|\mX_{a}^{(p)}\! -\! \mX_{a}^\star\|_F}, \mX_{a+1}^\star,  \dots, \! \mX_n^\star]  \bigg\> \nonumber\\
    &\!\!\!\!  \!\!\!\!& + \|\wh\vrho - \vrho^\star\|_F\sum_{q=1}^{Q}\sum_{k=1}^{K}\eta_{q,k} \< \mA_{q,k}, \vrho^{(p)}  \>\nonumber\\
    &\!\!\!\! \leq \!\!\!\!&  \|\wh\vrho - \vrho^\star\|_F n\epsilon \max_{\vrho\in\ol\setX_{\textup{MPO},2r^{\textup{MPO}}}}\sum_{q=1}^{Q}\sum_{k=1}^{K}\eta_{q,k} \< \mA_{q,k}, \vrho    \> + \|\wh\vrho - \vrho^\star\|_F\sum_{q=1}^{Q}\sum_{k=1}^{K}\eta_{q,k} \< \mA_{q,k}, \vrho^{(p)}  \>.
    \end{eqnarray}
    Setting $\epsilon=\frac{1}{2n}$ yields
    \begin{eqnarray}
    \label{upper bound of the right term MPO1}
    \sum_{q=1}^{Q}\sum_{k=1}^{K}\eta_{q,k} \< \mA_{q,k}, \wh\vrho - \vrho^\star   \> &\!\!\!\! \leq \!\!\!\!& \|\wh\vrho - \vrho^\star\|_F\max_{\vrho\in\ol\setX_{\textup{MPO},2r^{\textup{MPO}}}}\sum_{q=1}^{Q}\sum_{k=1}^{K}\eta_{q,k} \< \mA_{q,k}, \vrho  \>\nonumber\\
    &\!\!\!\! \leq \!\!\!\!& 2\|\wh\vrho - \vrho^\star\|_F\sum_{q=1}^{Q}\sum_{k=1}^{K}\eta_{q,k} \< \mA_{q,k}, \vrho^{(p)}  \>.
    \end{eqnarray}

\item \textbf{PEPOs} Let $\{\vrho^{(1)}, \dots, \vrho^{(N_{\textup{PEPO},2r^{\textup{PEPO}}})}\} \subset\ol\setX_{\textup{PEPO},2r^{\textup{PEPO}}}$   be the $\epsilon_{\ul{a}\,\ul{b}}$-net constructed in Appendix~\ref{Proof of ocvering number all} with $L=\|\wh\vrho-\vrho^\star\|_F$, and fix an arbitrary covering point $\vrho^{(p)}\in \{\vrho^{(1)}, \dots, \vrho^{(N_{\textup{PEPO},2r^{\textup{PEPO}}})}\}$. For notational convenience, we represent the PEPO corresponding to $\vrho$ compactly as  $\vrho = [\mX_{\ul{1}\,\ul{1}}, \dots, \mX_{\ul{q}\,\ul{p}} ] $, where $\mX_{\ul{a},\ul{b}}$ denotes the local tensor at site $(\ul{a},\ul{b})$. Applying the same argument as that used to derive the third inequality in \eqref{upper bound of the right term MPO}, we obtain
    \begin{eqnarray}
    \label{upper bound of the right term PEPO}
    &\!\!\!\! \!\!\!\!&\sum_{q=1}^{Q}\sum_{k=1}^{K}\eta_{q,k} \< \mA_{q,k}, \wh\vrho - \vrho^\star   \>\leq  \max_{\vrho\in\ol\setX_{\textup{MPO},2r^{\textup{PEPO}}}}\sum_{q=1}^{Q}\sum_{k=1}^{K}\eta_{q,k} \< \mA_{q,k}, \vrho  \>\nonumber\\
    &\!\!\!\! \leq \!\!\!\!&  \sum_{a=1}^{q}\sum_{b=1}^{p} L_{\ul{a},\ul{b}} \epsilon_{\ul{a},\ul{b}}\max_{[\mX_{\ul{1}\,\ul{1}}, \dots, \mX_{\ul{q}\,\ul{p}} ]\in\ol\setX_{\textup{PEPO},2r^{\textup{PEPO}}}}\sum_{q=1}^{Q}\sum_{k=1}^{K}\eta_{q,k} \bigg\< \mA_{q,k}, [\mX_{\ul{1}\,\ul{1}}^{(p)},\dots,\! \frac{\mX_{\ul{a}\,\ul{b}}^{(p)}\! -\! \mX_{\ul{a}\,\ul{b}}^\star}{\|\mX_{\ul{a}\,\ul{b}}^{(p)}\! -\! \mX_{\ul{a}\,\ul{b}}^\star\|_F},  \dots, \! \mX_{\ul{q}\,\ul{p}}^\star]  \bigg\> \nonumber\\
    &\!\!\!\!  \!\!\!\!& + \sum_{q=1}^{Q}\sum_{k=1}^{K}\eta_{q,k} \< \mA_{q,k}, \vrho^{(p)}  \>\nonumber\\
    &\!\!\!\! \leq \!\!\!\!&   \sum_{a=1}^{q}\sum_{b=1}^{p}  \epsilon_{\ul{a},\ul{b}} \max_{\vrho\in\ol\setX_{\textup{PEPO},2r^{\textup{PEPO}}}} \sum_{q=1}^{Q}\sum_{k=1}^{K}\eta_{q,k}\< \mA_{q,k}, \vrho   \> + \sum_{q=1}^{Q}\sum_{k=1}^{K}\eta_{q,k} \< \mA_{q,k}, \vrho^{(p)}  \>.
    \end{eqnarray}
    Setting $\epsilon_{\ul{a},\ul{b}}=\frac{1}{2n}$ yields
    \begin{eqnarray}
    \label{upper bound of the right term PEPO1}
    \sum_{q=1}^{Q}\sum_{k=1}^{K}\eta_{q,k} \< \mA_{q,k}, \wh\vrho - \vrho^\star   \> \leq 2 \sum_{q=1}^{Q}\sum_{k=1}^{K}\eta_{q,k} \< \mA_{q,k}, \vrho^{(p)}  \>.
    \end{eqnarray}

\end{itemize}

We consider any fixed value of $\vrho^{(p)}$ and apply \cite[Lemma 14]{qin2024quantum} to establish a concentration inequality for the expression $\sum_{q=1}^{Q}\sum_{k=1}^{K}  \veta_{q,k}\<\mA_{q,k}, \vrho^{(p)} \>$. Specifically, we have
\begin{eqnarray}
\label{concentration inequality of fixed quantum state}
\P{\sum_{q=1}^{Q}\sum_{k=1}^{K}  \eta_{q,k}\<\mA_{q,k}, \vrho^{(p)} \> > t}  \leq    2e^{-\frac{Mt^2}{16 \sum_{q = 1}^{Q}\sum_{k = 1}^{K}\<\mA_{q,k}, \vrho^{(p)} \>^2p_{q,k}}}.
\end{eqnarray}
Furthermore,
\begin{eqnarray}
\label{expansion of three order term}
&\!\!\!\!   \!\!\!\!&\sum_{q = 1}^{Q}\sum_{k = 1}^{K}\<\mA_{q,k}, \vrho^{(p)} \>^2p_{q,k} \nonumber\\
&\!\!\!\! = \!\!\!\!& \sum_{q = 1}^{Q}\sum_{k = 1}^{K}\<\mA_{q,k}, \vrho^{(p)} \>^2\bigg\<(1-a) \mA_{q,k} + \frac{a}{K}\mId,(1-\lambda)\vrho^\star + \frac{\lambda}{d^n}\mId \bigg\>\nonumber\\
&\!\!\!\! = \!\!\!\!& (1-a)(1-\lambda)\sum_{q = 1}^{Q}\sum_{k = 1}^{K}\<\mA_{q,k}, \vrho^{(p)} \>^2 \< \mA_{q,k},\vrho^\star  \>  + (1-a)\lambda\sum_{q = 1}^{Q}\sum_{k = 1}^{K}\<\mA_{q,k}, \vrho^{(p)} \>^2 \< \mA_{q,k},  \frac{1}{d^n}\mId  \>\nonumber\\
&\!\!\!\!   \!\!\!\!& + \frac{a(1-\lambda)}{K}\sum_{q = 1}^{Q}\sum_{k = 1}^{K}\<\mA_{q,k}, \vrho^{(p)} \>^2 \< \mId,  \vrho^\star   \>  + \frac{a\lambda}{Kd^n}\sum_{q = 1}^{Q}\sum_{k = 1}^{K}\<\mA_{q,k}, \vrho^{(p)} \>^2 \< \mId,  \mId \>\nonumber\\
&\!\!\!\! \leq \!\!\!\!& (1-a)\beta(Q,K)\|\vrho^{(p)}\|_F^2 + \frac{a\alpha_2(Q,K)}{K}\|\vrho^{(p)}\|_F^2,
\end{eqnarray}
where the four terms in the expansion are bounded separately: the first term by \eqref{upper constant of third order information}; the second term by the same inequality together with $\frac{1}{d^n}\mId\in\setX_{\textup{PHY}}$; and the last two terms by \eqref{upper constant of second order information}, using $\trace(\vrho^\star)=1$ and $\trace(\mId)=d^n$. Consequently, together with \eqref{concentration inequality of fixed quantum state}, and noting that $\|\vrho^{(p)}\|_F^2\le1$ for physical states, low-rank states, sparse states, and MPOs, while $\|\vrho^{(p)}\|_F^2\le\|\wh\vrho-\vrho^\star\|_F^2$ for PEPOs, we obtain
\begin{eqnarray}
\label{concentration inequality of fixed quantum state1}
&\!\!\!\! \!\!\!\!&\P{\sum_{q=1}^{Q}\sum_{k=1}^{K}\eta_{q,k} \< \mA_{q,k}, \wh\vrho - \vrho^\star   \> > t}\nonumber\\
&\!\!\!\!\leq \!\!\!\!& \begin{cases} \P{\sum_{q=1}^{Q}\sum_{k=1}^{K}  \eta_{q,k}\<\mA_{q,k}, \vrho^{(p)} \> > \frac{t}{2\|\wh\vrho-\vrho^\star\|_F}}, & \setX_{\textup{PHY}}, \setX_{\textup{S}}, \setX_{\textup{LR}}, \setF_{\textup{LR}}, \setX_{\textup{MPO}} \\
\P{\sum_{q=1}^{Q}\sum_{k=1}^{K}  \eta_{q,k}\<\mA_{q,k}, \vrho^{(p)} \> > \frac{t}{2}}, & \setX_{\textup{PEPO}}\\
\end{cases}\nonumber\\
&\!\!\!\!\leq \!\!\!\!& \begin{cases}
2N_{\textup{PHY}}e^{-\frac{Mt^2}{64((1-a)\beta(Q,K) + a\alpha_2(Q,K)/K)\|\wh\vrho-\vrho^\star\|_F^2}}, &  \setX_{\textup{PHY}} \\
2N_{\textup{S},2s}e^{-\frac{Mt^2}{64((1-a)\beta(Q,K) + a\alpha_2(Q,K)/K)\|\wh\vrho-\vrho^\star\|_F^2}}, &  \setX_{\textup{S}} \\
2N_{\textup{LR},2r^{\textup{LR}}}e^{-\frac{Mt^2}{64((1-a)\beta(Q,K) + a\alpha_2(Q,K)/K)\|\wh\vrho-\vrho^\star\|_F^2}}, &  \setX_{\textup{LR}}, \setF_{\textup{LR}}  \\
2N_{\textup{MPO},2r^{\textup{MPO}}}e^{-\frac{Mt^2}{64((1-a)\beta(Q,K) + a\alpha_2(Q,K)/K)\|\wh\vrho-\vrho^\star\|_F^2}}, &  \setX_{\textup{MPO}}  \\
2N_{\textup{PEPO},2r^{\textup{PEPO}}}e^{-\frac{Mt^2}{64((1-a)\beta(Q,K) + a\alpha_2(Q,K)/K)\|\wh\vrho-\vrho^\star\|_F^2}}, &  \setX_{\textup{PEPO}}  \\
\end{cases}\nonumber\\
&\!\!\!\!\leq \!\!\!\!& e^{-\frac{CMt^2}{((1-a)\beta(Q,K) + a\alpha_2(Q,K)/K)\|\wh\vrho-\vrho^\star\|_F^2}+\log N_{\textup{type}}},    \setX_{\textup{PHY}}, \setX_{\textup{S}}, \setX_{\textup{LR}}, \setF_{\textup{LR}}, \setX_{\textup{MPO}},\setX_{\textup{PEPO}}.
\end{eqnarray}
Here $N_{\textup{PHY}}$, $N_{\textup{S},2s}$, $N_{\textup{LR},2r^{\textup{LR}}}$, $N_{\textup{MPO},2r^{\textup{MPO}}}$ and $N_{\textup{PEPO},2r^{\textup{PEPO}}}$ denote the corresponding covering numbers, whose bounds are established in Appendix~\ref{Proof of ocvering number all}. We choose $\epsilon=\frac12$ for physical states, $\epsilon=\frac14$ for sparse and low-rank states, and $\epsilon=\frac{1}{2n}$ for MPOs. For PEPOs, we choose $\epsilon_{\ul{a},\ul{b}}=\frac{1}{2n}$. In the last inequality, $C>0$ denotes a universal constant, and $N_{\textup{type}}$ represents the covering number associated with the underlying state class given in \Cref{tab:covering-numbers}.

Through taking
\begin{eqnarray}
\hat t = O\bigg(\sqrt{\frac{((1-a)\beta(Q,K) + a\alpha_2(Q,K)/K)(\log N_{\textup{type}})}{M}}\bigg) \cdot \|\wh\vrho-\vrho^\star\|_F, \  \setX_{\textup{PHY}}, \setX_{\textup{S}}, \setX_{\textup{LR}}, \setF_{\textup{LR}}, \setX_{\textup{MPO}},  \setX_{\textup{PEPO}},
\end{eqnarray}
with probability $1 - e^{-\Omega(\log N_{\textup{type}})}$, we further obtain
\begin{eqnarray}
\label{concentration inequality of fixed quantum state2}
    \sum_{q=1}^{Q}\sum_{k=1}^{K}\eta_{q,k} \< \mA_{q,k}, \wh\vrho - \vrho^\star \> \leq O\bigg(\sqrt{\frac{((1-a)\beta(Q,K) + a\alpha_2(Q,K)/K)(\log N_{\textup{type}})}{M}}\bigg) \cdot \|\wh\vrho-\vrho^\star\|_F.
\end{eqnarray}
Combining \eqref{necessary condition of global solution1}, \eqref{lower bound of the left term}, \eqref{upper bound of the right term} and \eqref{concentration inequality of fixed quantum state2}, we obtain
\begin{eqnarray}
\label{upper bound of solution for prior knowledge}
    \|\wh\vrho - \vrho^\star\|_F \leq O\bigg(\sqrt{\frac{((1-a)\beta(Q,K) + a\alpha_2(Q,K)/K)(\log N_{\textup{type}})}{(\alpha_1(Q,K))^2 (1-a)^2(1-\lambda)^2M}}\bigg), \ \ \setX_{\textup{PHY}}, \setX_{\textup{S}}, \setX_{\textup{LR}}, \setF_{\textup{LR}}, \setX_{\textup{MPO}},\setX_{\textup{PEPO}}.
\end{eqnarray}

We next extend the analysis to the remaining structured quantum-state classes. Specifically, we consider
\begin{itemize}

\item \textbf{Low-rank parameterizations of sparse states, MPOs, and PEPOs} Let $\mF$ belong to one of the sets $\setF_{\textup{LR-S}}$, $\setF_{\textup{LR-MPO}}$, or $\setF_{\textup{LR-PEPO}}$, and define $\vrho=\mF\mF^\dagger$. Then $\vrho$ can be regarded as an element of $\setX_{\textup{S}}$ with sparsity $t^2$, $\setX_{\textup{MPO}}$ with bond dimensions $[(r_1^{\textup{LR-MPO}})^2, \dots, (r_{n-1}^{\textup{LR-MPO}})^2 ]$, or $\setX_{\textup{PEPO}}$ with bond dimensions $[ (r_{\ul{1}\,\ul{1},\,\ul{1}\,\ul{2}}^{\textup{LR-PEPO}})^2,\dots,  (r_{\ul{q}\,\ul{p-1},\,\ul{q}\,\ul{p}}^{\textup{LR-PEPO}})^2 ]$, respectively. Consequently, applying \eqref{upper bound of solution for prior knowledge} with the corresponding structural parameters yields that, with probability $1 - e^{-\Omega(\log N_{\textup{type}})}$, we have
    \begin{eqnarray}
    \label{upper bound of solution for prior knowledge other three}
    \|\wh\vrho - \vrho^\star\|_F \leq O\bigg(\sqrt{\frac{((1-a)\beta(Q,K) + a\alpha_2(Q,K)/K)(\log N_{\textup{type}})}{(\alpha_1(Q,K))^2 (1-a)^2(1-\lambda)^2M}}\bigg), \ \ \setF_{\textup{LR-S}}, \setF_{\textup{LR-MPO}}, \setF_{\textup{LR-PEPO}},
    \end{eqnarray}
    where $\log N_{\textup{type}}$ is defined in \Cref{tab:covering-numbers}.

\item\textbf{MPSs and PEPSs} Finally, we consider the structured classes $\setF_{\textup{MPS}}$ and $\setF_{\textup{PEPS}}$. Although both MPSs and PEPSs can be regarded as special cases of MPOs and PEPOs, directly applying the preceding analysis would not fully exploit their additional low-dimensional structures. Instead, we derive sharper recovery guarantees by relating the stochastic term $\sum_{q=1}^{Q}\sum_{k=1}^{K}\eta_{q,k} \< \mA_{q,k}, \wh\vrho - \vrho^\star   \> = \sum_{q=1}^{Q}\sum_{k=1}^{K}\eta_{q,k} \< \mA_{q,k}, \wh\vf\wh\vf^\dagger - \vf^\star{\vf^\star}^\dagger   \>$ to the vector estimation error $\|\wh\vf - \vf^\star\|_2$. Recall the auxiliary parameter sets introduced in Appendix~\ref{Proof of ocvering number all}, $\ol\setF_{\textup{MPS},r^{\textup{MPS}}}$ and $\ol\setF_{\textup{PEPS},r^{\textup{PEPS}}}$, defined as $\ol\setF_{\textup{MPS},r^{\textup{MPS}}}  = \{ \vf\in\C^{d^n\times 1}:\  \|\vf\|_2 = 1, \vf(i_1 \cdots i_\nqbit) = \mX_1^{i_1}  \cdots \mX_\nqbit^{i_\nqbit}, \mX_\ell^{i_\ell}\in\C^{r_{\ell-1}^{\textup{MPS}}\times r_\ell^{\textup{MPS}}}, \ell\in[n], r_0^{\textup{MPS}}=r_n^{\textup{MPS}}=1  \}$ and $\ol\setF_{\textup{PEPS},r^{\textup{PEPS}}} = \{ \vf\in\C^{d^n\times 1}: \vf = [\mX_{\ul{1}\,\ul{1}}, \dots, \mX_{\ul{q}\,\ul{p}} ], \|\mX_{\ul{a}\,\ul{b}} \|_F\leq L_{\ul{a}\,\ul{b}}, \mX_{\ul{a}\,\ul{b}}\in\C^{d\times r_{\ul{a}\,\ul{b-1},\,\ul{a}\,\ul{b}}^{\textup{PEPS}}\times r_{\ul{a-1}\,\ul{b},\,\ul{a}\,\ul{b}}^{\textup{PEPS}} \times r_{\ul{a}\,\ul{b},\,\ul{a}\,\ul{b+1}}^{\textup{PEPS}}\times r_{\ul{a}\,\ul{b},\,\ul{a+1}\,\ul{b}}^{\textup{PEPS}}}, \ul{a}\in[q],  \ul{b}\in[p]  \} \cap \{\vf\in\C^{d^n\times 1}: \|\vf\|_2\leq \|\wh\vf - \vf^\star\|_2  \}$. Specifically, the stochastic term admits the following bound:
    \begin{eqnarray}
    \label{upper bound of the right term MPS and PEPS}
    \sum_{q=1}^{Q}\sum_{k=1}^{K}\eta_{q,k} \< \mA_{q,k}, \wh\vrho - \vrho^\star   \> \leq \begin{dcases}
    \|\wh\vf - \vf^\star\|_2 \cdot \max_{\vf\in\overline\setF_{\textup{MPS},2r^{\textup{MPS}}}} \sum_{q=1}^{Q}\sum_{k=1}^{K}\eta_{q,k} \< \mA_{q,k}, \wh\vf\vf^\dagger + \vf{\vf^\star}^\dagger   \>, &  \text{MPS}, \\
    \max_{\vf\in\overline\setF_{\textup{PEPS},2r^{\textup{PEPS}}}} \sum_{q=1}^{Q}\sum_{k=1}^{K}\eta_{q,k} \< \mA_{q,k}, \wh\vf\vf^\dagger + \vf{\vf^\star}^\dagger   \>, &  \text{PEPS}.
    \end{dcases}
    \end{eqnarray}
    Applying the same concentration argument as in \eqref{upper bound of the right term MPO} and \eqref{upper bound of the right term PEPO}, we obtain that, with probability at least $1 - e^{-\Omega(\log N_{\textup{type}})}$, we further obtain
    \begin{eqnarray}
    \label{concentration inequality of MPS PEPS}
    &\!\!\!\!\!\!\!\!&\sum_{q=1}^{Q}\sum_{k=1}^{K}\eta_{q,k} \< \mA_{q,k}, \wh\vrho - \vrho^\star \>\nonumber\\
    &\!\!\!\!\leq\!\!\!\!& O\bigg(\sqrt{\frac{((1-a)\beta(Q,K) + a\alpha_2(Q,K)/K)(\log N_{\textup{type}})}{M}}\bigg) \cdot \|\wh\vf - \vf^\star\|_2, \ \setF_{\textup{MPS}}, \setF_{\textup{PEPS}}.
    \end{eqnarray}
    Combining \eqref{necessary condition of global solution1}, \eqref{lower bound of the left term}, \eqref{upper bound of the right term}, together with the relation $\|\wh\vf - \vf^\star\|_2\leq \frac{1}{2(\sqrt{2}-1)}\|\wh{\vrho} - \vrho^\star\|_F$ established in \cite[Lemma 41]{ge2017no}, yields
    \begin{eqnarray}
    \label{upper bound of solution for prior knowledge MPS PEPS}
    \|\wh\vrho - \vrho^\star\|_F \leq O\bigg(\sqrt{\frac{((1-a)\beta(Q,K) + a\alpha_2(Q,K)/K)(\log N_{\textup{type}})}{(\alpha_1(Q,K))^2 (1-a)^2(1-\lambda)^2M}}\bigg), \ \ \setF_{\textup{MPS}}, \setF_{\textup{PEPS}}.
    \end{eqnarray}
\end{itemize}

By applying the inequality $\|\wh{\vrho} - \vrho^\star\|_1 \leq 2\sqrt{\rank(\vrho^\star)}\|\wh{\vrho} - \vrho^\star\|_F$~\cite{coles2019strong}, we can derive
    \begin{eqnarray}
    \label{upper bound of solution for prior knowledge summary appendix}
    \|\wh\vrho - \vrho^\star\|_1 \leq O\bigg(\sqrt{\frac{((1-a)\beta(Q,K) + a\alpha_2(Q,K)/K)\rank(\vrho^\star)(\log N_{\textup{type}})}{(\alpha_1(Q,K))^2 (1-a)^2(1-\lambda)^2M}}\bigg).
    \end{eqnarray}
This completes the proof.

\section{Proof of \Cref{label:sample complexity of biased solution}}
\label{sec: proof of sample complexity of bisaed solution}

The proof follows the same argument as that of Appendix~\ref{sec: proof of sample complexity of unbisaed solution}. Since $\wh\vrho$ is a global minimizer of \eqref{loss function of biased solution}, it satisfies
\begin{eqnarray}
\label{necessary condition of global solution biased}
g(\wh\vrho) \leq g(\vrho^\star).
\end{eqnarray}
Expanding \eqref{necessary condition of global solution biased} and substituting $\wh p_{q,k}=p_{q,k}+\eta_{q,k}$ yield
\begin{eqnarray}
\label{necessary condition of global solution1 biased}
\frac{1}{2QK}\sum_{q=1}^{Q}\sum_{k=1}^{K} \< \mA_{q,k},  \wh\vrho - \vrho^\star   \>^2 \leq \frac{1}{QK} \sum_{q=1}^{Q}\sum_{k=1}^{K}\eta_{q,k}\< \mA_{q,k},   \wh\vrho - \vrho^\star  \> + \frac{1}{QK} \sum_{q=1}^{Q}\sum_{k=1}^{K} L_{q,k}\< \mA_{q,k},   \wh\vrho - \vrho^\star  \>
\end{eqnarray}
with
\begin{eqnarray}
\label{definition of Lqk}
L_{q,k} = \bigg\<\mA_{q,k},  -\lambda\vrho^\star + \frac{\lambda}{d^n}\mId \bigg\> + \bigg\< -a\mA_{q,k} + \frac{a}{K}\mId,   \vrho^\star   \bigg\> + \bigg\<a\mA_{q,k} - \frac{a}{K}\mId,  \lambda\vrho^\star - \frac{\lambda}{d^n}\mId \bigg\>.
\end{eqnarray}
Applying \eqref{lower constant of second order information} yields
\begin{eqnarray}
\label{lower bound of left term biased}
\frac{1}{2QK}\sum_{q=1}^{Q}\sum_{k=1}^{K} \< \mA_{q,k},  \wh\vrho - \vrho^\star   \>^2 \geq \frac{\alpha_1(Q,K)}{2QK}   \|\wh\vrho - \vrho^\star\|_F^2.
\end{eqnarray}
The first term on the right-hand side of \eqref{necessary condition of global solution1 biased} can be bounded by applying the same argument as in Appendix~\ref{sec: proof of sample complexity of unbisaed solution}. Consequently, with probability at least $1-e^{-\Omega(\log N_{\textup{type}})}$, we obtain
\begin{eqnarray}
\label{upper bound of first right term biased}
    \sum_{q=1}^{Q}\sum_{k=1}^{K}\eta_{q,k} \< \mA_{q,k}, \wh\vrho - \vrho^\star \> \leq O\bigg(\sqrt{\frac{((1-a)\beta(Q,K) + a\alpha_2(Q,K)/K)(\log N_{\textup{type}})}{M}}\bigg) \cdot \|\wh\vrho-\vrho^\star\|_F.
\end{eqnarray}
For MPSs and PEPSs, the above bound follows by combining \eqref{concentration inequality of MPS PEPS} with the relation $\|\wh\vf - \vf^\star\|_2\leq \frac{1}{2(\sqrt{2}-1)}\|\wh{\vrho} - \vrho^\star\|_F$ established in \cite[Lemma 41]{ge2017no}.

Next, we analyze the second term on the right-hand side of \eqref{necessary condition of global solution1 biased}.
\begin{eqnarray}
\label{upper bound of second right term biased}
    &\!\!\!\!\!\!\!\!&\sum_{q=1}^{Q}\sum_{k=1}^{K} L_{q,k}\< \mA_{q,k},   \wh\vrho - \vrho^\star  \>\nonumber\\
    &\!\!\!\! = \!\!\!\!&  \sum_{q=1}^{Q}\sum_{k=1}^{K}\bigg\<\mA_{q,k},  -\lambda\vrho^\star + \frac{\lambda}{d^n}\mId \bigg\>\< \mA_{q,k},   \wh\vrho - \vrho^\star  \> + \bigg\< -a\mA_{q,k} + \frac{a}{K}\mId,   \vrho^\star   \bigg\>\< \mA_{q,k},   \wh\vrho - \vrho^\star  \>\nonumber\\
    &\!\!\!\!\!\!\!\!& + \bigg\<a\mA_{q,k} - \frac{a}{K}\mId,  \lambda\vrho^\star - \frac{\lambda}{d^n}\mId \bigg\>\< \mA_{q,k},   \wh\vrho - \vrho^\star  \>\nonumber\\
    &\!\!\!\! = \!\!\!\!&\sum_{q=1}^{Q}\sum_{k=1}^{K}\bigg\<\mA_{q,k},  -\lambda\vrho^\star + \frac{\lambda}{d^n}\mId \bigg\> \< \mA_{q,k},   \wh\vrho - \vrho^\star   \> - a \<\mA_{q,k},  \vrho^\star  \> \< \mA_{q,k},   \wh\vrho - \vrho^\star   \>\nonumber\\
    &\!\!\!\!\!\!\!\!& + \bigg\<a\mA_{q,k},  \lambda\vrho^\star - \frac{\lambda}{d^n}\mId \bigg\>\< \mA_{q,k},   \wh\vrho - \vrho^\star  \>\nonumber\\
    &\!\!\!\! \leq \!\!\!\!& \gamma_2(Q,K)(1-a)\trace\bigg((\wh\vrho - \vrho^\star) \bigg(\frac{\lambda}{d^n}\mId - \lambda\vrho^\star\bigg) \bigg) - a\gamma_1(Q,K)\trace( \vrho^\star(\wh\vrho - \vrho^\star))\nonumber\\
    &\!\!\!\! = \!\!\!\!& -\lambda\gamma_2(Q,K)(1-a)\trace( \vrho^\star(\wh\vrho - \vrho^\star)) - a\gamma_1(Q,K)\trace( \vrho^\star(\wh\vrho - \vrho^\star))\nonumber\\
    &\!\!\!\! \leq \!\!\!\!& \big( \lambda(1-a)\gamma_2(Q,K) + a\gamma_1(Q,K) \big)\|\wh\vrho - \vrho^\star\|_F,
\end{eqnarray}
where the second equality follows from $\sum_{q=1}^{Q}\sum_{k=1}^{K} \< \frac{a}{K}\mId,   \vrho^\star   \>\< \mA_{q,k},   \wh\vrho - \vrho^\star  \> = \< \frac{aQ}{K}\mId,   \vrho^\star   \>\trace(\wh\vrho - \vrho^\star) = 0$ and $\< \frac{a}{K}\mId,  \lambda\vrho^\star - \frac{\lambda}{d^n}\mId \> = \frac{a\lambda}{K}\trace(\vrho^\star - \frac{1}{d^n}\mId) =0$. The first inequality follows from the definitions of
$\gamma_1(Q,K)$ and $\gamma_2(Q,K)$ in \eqref{lower constant of second order information cross term} and \eqref{upper constant of second order information cross term}. The second inequality follows from $-\trace(\vrho^\star(\wh\vrho-\vrho^\star)) \leq |\trace( \vrho^\star(\wh\vrho - \vrho^\star))|\leq \|\vrho^\star\|_F\|\wh\vrho - \vrho^\star\|_F\leq \trace(\vrho^\star)\|\wh\vrho - \vrho^\star\|_F\leq \|\wh\vrho - \vrho^\star\|_F$.

Combining \eqref{necessary condition of global solution1 biased}, \eqref{lower bound of left term biased}, \eqref{upper bound of first right term biased} with \eqref{upper bound of second right term biased}, we have
\begin{eqnarray}
\label{biased error bound of trace norm}
    \|\wh{\vrho} - \vrho^\star\|_1 &\!\!\!\!\leq\!\!\!\!& 2\sqrt{\rank(\vrho^\star)}\|\wh{\vrho} - \vrho^\star\|_F\nonumber\\
    &\!\!\!\!\leq\!\!\!\!& O\bigg(\sqrt{\frac{((1-a)\beta(Q,K) + a\alpha_2(Q,K)/K)\rank(\vrho^\star)(\log N_{\textup{type}})}{(\alpha_1(Q,K))^2 M}}\bigg)\nonumber\\
    &\!\!\!\!\!\!\!\!& + \frac{4\sqrt{\rank(\vrho^\star)}\big( \lambda(1-a)\gamma_2(Q,K) + a\gamma_1(Q,K) \big)}{\alpha_1(Q,K)},
\end{eqnarray}
where the first inequality uses $\|\wh{\vrho} - \vrho^\star\|_1 \leq 2\sqrt{\rank(\vrho^\star)}\|\wh{\vrho} - \vrho^\star\|_F$~\cite{coles2019strong}. This completes the proof.


\end{document}